\documentclass[12pt]{article}
\usepackage{fullpage}

\usepackage{bm}
\usepackage[dvipsnames]{xcolor}
\usepackage{amsmath,amssymb}\usepackage{authblk}
\usepackage{amsthm}
\usepackage{wrapfig}
\usepackage{epsfig}
\usepackage{graphicx}
\usepackage{makeidx}
\usepackage{multicol}
\usepackage{tikz}
\usepackage{algorithm}
\usepackage{algpseudocode}
\usepackage{soul}

\usepackage{array,multirow,booktabs,bbm} 

\usepackage{algorithmicx}
\usepackage{algpseudocode}
\makeatletter
\algrenewcommand\ALG@beginalgorithmic{\sffamily\small}
\makeatother
\usepackage{epstopdf}
\usepackage{mathrsfs,mathtools}
\usepackage{xcolor}
\usepackage{enumitem}
\usepackage{comment}

\newcommand{\bs}[1]{\boldsymbol{#1}}

\def \ri {{\rm i}}

\usepackage{enumerate}

\newtheorem{proposition}{Proposition}[section]
\newtheorem{theorem}{Theorem}[section]

\theoremstyle{remark}
\newtheorem{remark}{Remark}[section]

\numberwithin{equation}{section}
\numberwithin{figure}{section}
\numberwithin{table}{section}
\numberwithin{algorithm}{section}

\usepackage[title]{appendix}

\usepackage[colorlinks,linkcolor=red]{hyperref}

\graphicspath{{./Figs/}}

\usepackage{multirow}

\allowdisplaybreaks
\usepackage{algorithm,algpseudocode,float}

\usepackage{lipsum}
\makeatletter
\newenvironment{breakablealgorithm}
  {
   \begin{center}
     \refstepcounter{algorithm}
     \hrule height.8pt depth0pt \kern2pt
     \renewcommand{\caption}[2][\relax]{
       {\raggedright\textbf{\ALG@name~\thealgorithm} ##2\par}%
       \ifx\relax##1\relax 
         \addcontentsline{loa}{algorithm}{\protect\numberline{\thealgorithm}##2}%
       \else 
         \addcontentsline{loa}{algorithm}{\protect\numberline{\thealgorithm}##1}%
       \fi
       \kern2pt\hrule\kern2pt
     }
  }{
     \kern2pt\hrule\relax
   \end{center}
  }
\makeatother

\usepackage{longtable}

\usepackage{subfig}
\usepackage{float}
\usepackage{caption}

\usepackage{hyperref}
\usepackage{marvosym}

\newcommand\keywords[1]{\textbf{Keywords}: #1}

\title{An Energy-Conserving Fourier Particle-in-Cell Method with Asymptotic-Preserving Preconditioner for Vlasov-Amp\`ere System with Exact Curl-Free Constraint
}
\author{Zhuoning Li
\footnote{ \tt{hkht0526@sjtu.edu.cn}}}
\author{Zhenli Xu
\footnote{ \tt{xuzl@sjtu.edu.cn}}}
\author{Zhiguo Yang
\footnote{ \tt{yangzhiguo@sjtu.edu.cn}}}

\begin{footnotesize}
\affil{School of Mathematical Sciences, MOE-LSC and CMA-Shanghai, Shanghai Jiao Tong University, Shanghai 200240, China.} 
\end{footnotesize}
\date{}

\begin{document}

\maketitle

\begin{abstract}
We present an efficient and accurate energy-conserving implicit particle-in-cell~(PIC) algorithm for the electrostatic Vlasov system, with particular emphasis on its high robustness for simulating complex plasma systems with multiple physical scales. This method consists of several indispensable elements: (\romannumeral1) the reformulation of the original Vlasov-Poisson system into an equivalent Vlasov-Amp\`ere system with divergence-/curl-free constraints; (\romannumeral2) a novel structure-preserving Fourier spatial discretization, which exactly preserves these constraints at the discrete level; (\romannumeral3) a preconditioned Anderson-acceleration algorithm for the solution of the highly nonlinear system; and (\romannumeral4) a linearized and uniform approximation of the implicit Crank-Nicolson scheme for various Debye lengths, based on the generalized Ohm's law, which serves as an asymptotic-preserving preconditioner for the proposed method. Numerical experiments are conducted, and comparisons are made among the proposed energy-conserving scheme, the classical leapfrog scheme, and a Strang operator-splitting scheme to demonstrate the superiority of the proposed method, especially for plasma systems crossing physical scales.
\end{abstract}

\keywords{Vlasov-Amp\`ere system, energy conservation, structure-preserving Fourier method, asymptotic-preserving preconditioner, operator splitting method}

\section{Introduction}

The Vlasov equation is a fundamental kinetic model of collisionless plasmas, which describes the evolution of the probability distribution function of electrically charged particles in six-dimensional phase space under self-induced and/or externally imposed electromagnetic fields \cite{10.1088/978-0-7503-1200-4}. This self-consistent coupling between the Vlasov equation and Maxwell or Poisson equation through the charge density and current density terms is highly nonlinear, especially when multiple physical and time scales co-exist in the system.

The high dimensionality, nonlinearity, multi-scale nature and various mathematical structures and properties of the system pose formidable challenges in the numerical simulation of plasma kinetic models \cite{MARKIDIS20101509,Pezzi2019ViDAAV}. One widely-used method to overcome the ``curse of dimensionality" issue is the particle-in-cell~(PIC) method  \cite{langdon1992enforcing,lapenta2015kinetic,Muga1970,VILLANI200271,birdsall2018plasma,hockney2021computer,dawson1983particle,2008Computational}. It approximates the Vlasov equation by Newton's second law of motion for a sequence of macro particles, and the interplay between macro particles and the electromagnetic fields is through the Lorenz force, charge density and current density terms, calculated via particle-grid interpolations and projections \cite{2017SMILEI}.

Nowadays, much effort have been made to develop structure-preserving PIC methods that can preserve the inherent physical properties and mathematical structures of the plasma system, such as charge \cite{crouseilles_respaud_2011,1999Exact,VILLASENOR1992306}, momentum \cite{MASON1981233,RICCI2002117,COHEN198215} and energy \cite{doi:10.1063/1.4979989,CHACON2016578} conservations, curl-free constraint of the electric field~(under the zero-magnetic limit) and divergence-free constraint of the magnetic field \cite{kraus2017gempic,campos2022variational}, and the Hamiltonian structure of the physical system \cite{MORRISON1980383,LI2019381,GU2022111472}. Among them, special attention is paid to the energy-conserving schemes \cite{lapenta2011particle,CHEN20142391,GONZALEZHERRERO2018162}, as they can effectively overcome finite-grid instability \cite{COHEN1989151,HEWETT1987121} and mitigate the particle self-heating or self-cooling issues \cite{HOCKNEY197119}. There is a vast amount of literature on energy-conserving schemes for the Vlasov-Maxwell~(VM) system \cite{chen2020semi,MARKIDIS20117037,lapenta2011particle,CHEN2019632,Ji2022AnAA} but relatively few studies for Vlasov-Poisson~(VP) equations \cite{chen2011energy}. This is mainly because Gauss's law in the VP system plays the role of a constraint, and the electric field responds instantaneously to charge density \cite{langdon1992enforcing}. Chen \emph{et al.} \cite{chen2011energy} proposed an energy-conserving fully-implicit PIC scheme for the one-dimensional VP system, in which the critical step is to rewrite the VP system into an equivalent Vlasov-Amp\`ere~(VA) system. Nevertheless, this technique could not be directly extended to a high-dimensional case. In two and three dimensions, the electric field in the Maxwell-Amp\`ere equation is required to be irrotational, and an artificial divergence-free variable needs to be introduced to guarantee the equivalence of the VP/VA systems. The curl-free and divergence-free constraints must be preserved exactly at the discrete level to ensure energy conservation. The first contribution of this paper is to propose a structure-preserving Fourier discretization method that precisely preserves these constraints, which, together with the time-centred Crank-Nicolson~(CN) scheme \cite{CN1996,Sun2004589}, leads to an energy-conserving scheme for the high-dimensional VA system.

In the numerical simulation of plasma systems, one primary and challenging issue is handling quasi-neutrality. Specifically, when the Debye length and plasma period are small compared to the space and time scales, it is called a quasi-neutral system \cite{fushion1974,plasma1986}. Though the proposed fully-implicit PIC scheme preserves the total energy, as the problem approaches the quasi-neutral limit, the nonlinear coupling dominates, and the resultant system becomes increasingly difficult to solve. The second contribution of this paper is to propose an asymptotic-preserving preconditioner, which is a linearized and uniform approximation concerning the Debye length for the fully-implicit scheme and remains non-degenerate even if the Debye length goes to zero. This idea originates from the asymptotic-preserving reformulation proposed by Degond \emph{et al.} \cite{degond2017asymptotic}~(see \cite{JinAP,Jinreview} for a review of the asymptotic-preserving schemes), where a generalized Ohm's law obtained from the Vlasov equation is used to discretize the current density in the Amp\`ere equation. The fully-implicit energy-conserving scheme, in combination with the asymptotic-preserving preconditioner and a preconditioned Anderson-acceleration algorithm, significantly improves the computational efficiency of the proposed method. In addition to the fully-implicit method,  we offer an energy conservation algorithm that decouples the updation of particle positions from the solution of particle velocities and electromagnetic fields. This is obtained with the help of operator splitting \cite{doi:10.1137/120871791,2022Zhong,1968Strang} and the structure-preserving Fourier discretizations \cite{Ameres2018Stochastic}. Comparisons are made among the proposed methods and the classical leapfrog scheme to demonstrate the scope of applicability of these methods.

The outline of this paper is as follows. In Section \ref{sec: model}, the reformulation of the Vlasov-Poisson system into the Vlasov-Amp\`ere system with divergence-free and curl-free constraints is introduced and the equivalence of the VP and VA models is verified. Section \ref{sec: div-curl-free-expansion} is denoted to the structure-preserving Fourier discretizations, which guarantee exact preservations of these constraints. A fully-implicit energy-conserving scheme with an asymptotic-preserving preconditioner and its solution algorithm based on Anderson-acceleration method are proposed in Section \ref{sec: fully-implicit}. Section \ref{sec: strang-splitting} is for an energy-conserving Strang operator-splitting scheme, which further reduces the computational cost of the fully-implicit scheme. Various numerical experiments are conducted in Section \ref{sec: numerical} to show the accuracy, efficiency and robustness of the proposed method. Finally, we conclude in Section \ref{sec: conclusion} with some closing remarks.

\section{Electrostatic Vlasov-Poisson/Amp\`ere system}
\label{sec: model}

Under the zero-magnetic field limit, a collisionless plasma is often described by the Vlasov-Poisson~(VP) system as follows:
\begin{subequations}
\begin{align}
    \displaystyle &\partial_{t}f_{s}(\bs{x},\bs{v},t)+\bs{v}\cdot\nabla_{\bs{x}}f_{s}(\bs{x},\bs{v},t)+\frac{q_{s}}{m_{s}}\bs{E}(\bs{x},t)\cdot\nabla_{\bs{v}}f_{s}(\bs{x},\bs{v},t)=0, \label{eq: classical-VP-Vlasov}\\
    \displaystyle &\nabla\cdot\bs{E}(\bs{x},t)=\frac{\rho(\bs{x},t)}{\epsilon_{0}},  \label{eq: classical-VP-Gauss}\\
    \displaystyle &\bs{E}(\bs{x},t)=-\nabla\phi(\bs{x},t). \label{eq: classical-VP-phi}
\end{align}
\label{eq: classical-VP}%
\end{subequations}
Here $f_{s}(\bs{x},\bs{v},t)$ is the distribution function of particles of species $s=1,\ldots,n$ at position $\bs x\in \Omega_{\bs x}$ with velocity $\bs v\in \Omega_{\bs{v}}$ at time $t\in \mathbb{R}^{+}$, $\bs{E}(\bs x,t)$ and $\phi(\bs x,t)$ are the electric field and electric potential, respectively, $\epsilon_{0}$ is the dielectric permittivity in vacuum, $q_{s}$ and $m_{s}$  are the valence and mass of particles of species $s$.
Number density of particles of species $s$, charge density $\rho(\bs x,t)$ and current density
$\bs J(\bs x,\bs{v},t)$ are defined as
\begin{equation}\label{eq: rhoj}
n_s(\bs{x},t)=\int_{\Omega_{\bs{v}}}f_{s}d\bs{v},\quad \rho(\bs{x},t)=\sum\limits_{s}q_{s}\int_{\Omega_{\bs{v}}}f_{s}d\bs{v}, \quad  \bs{J}(\bs{x},\bs{v},t)=\sum\limits_{s}q_{s}\int_{\Omega_{\bs{v}}}f_{s}\bs{v}d\bs{v}.
\end{equation}
By taking the integral of the Vlasov equation \eqref{eq: classical-VP-Vlasov} with respect to the velocity field and summing over $s$, $\rho(\bs x,t)$ and $\bs J(\bs x,\bs{v},t)$ are related to each other by the resultant charge continuity equation
\begin{equation}\label{eq: chargeeq}
    \displaystyle \frac{\partial\rho(\bs x,t)}{\partial t}+\nabla\cdot\bs{J}(\bs x,\bs{v},t)=0.
\end{equation}
The VP system is supplemented with the following initial condition
\begin{equation}\label{eq: initial}
 f_{s,0}(\bs x,\bs v)=f_s(\bs{x},\bs{v},0),\quad \bs E_0(\bs x)=\bs E(\bs x,0),
\end{equation}
and the given initial data needs to satisfy the compatible condition for the well-posedness of the problem
\begin{equation}\label{eq: wellp}
    \displaystyle \nabla\cdot\bs{E}_0(\bs{x})=\frac{\rho_0}{\epsilon_0}=\frac{1}{\epsilon_{0}} \sum\limits_{s}q_{s}\int_{\Omega_{\bs{v}}}f_{s,0}(\bs x,\bs v)d\bs{v}.
\end{equation}
For simplicity, we assume that $f_s$ and $\bs E$ satisfy the periodic boundary condition.

Without loss of generality, from now on we consider the case where ions with unit positive charge $e>0$ form a homogeneous motionless background, and electrons with charge $-e$ are the only species in this system.  Define the Debye length $\lambda_D$ and the electron plasma period $\tau_{p}$ by
\begin{equation}
\lambda_{D}=\sqrt{\frac{\epsilon_{0}k_B T_e}{e^{2}n_{e}}}  ,\quad  \tau_{p}=\sqrt{\frac{m_{e}\epsilon_{0}}{e^{2}n_{e}}},
\end{equation}
where $k_B$ is the Boltzmann constant, $T_e$ is the temperature, $m_e$ and $n_e$ are the mass and number density of electrons, respectively. We briefly comment on the non-dimensionalization procedure, which has been addressed in detail in \cite{degond2017asymptotic}. Let $x_{0}$ denote the length scale, $t_0$ the time scale, $E_0$ the electric intensity scale, and $n_0$ the number density scale. Besides, the charge and mass of particles are normalized by unit charge $e$ and mass of electron $m_{e}$ individually, and normalization parameter of velocity is $v_{0}=x_{0}/t_{0}$. In the case of uniformly-distributed ions, we take $n_{0}=n_{i}$, with $n_{i}$ the number density of ions.

\begin{table}[H]
 	\centering
     \begin{tabular}{c |c| c |c}
     \toprule[1.5pt]
       variables/parameters & normalization  &variables/parameters  & normalization   \\\midrule[1pt]
       $\bs x$   & $x_0$ & $t$ & $t_0$\\
       $\bs v$ & $v_0$ & $f$ & $n_0/v_0$\\
       $n$ & $n_0$ & $q$ & $q_e$\\
       $\bs E$ & $E_0$ & $\bs J$ & $en_{0}v_{0}$\\
       $m$& $m_{e}$ & $\rho$ & $en_{0}$\\
       \bottomrule[1.5pt]
     \end{tabular}
     \caption{Normalization of variables and simulation parameters}
     \label{tab: normalization}
 \end{table}

 By consistently normalizing the physical variables and parameters using Table \ref{tab: normalization} and assuming ${ex_0E_0}/{(m_e v_0^2)}=1$, $v_{0}B_{0}/E_{0}=1$, $v_{0}/v_{th,0}=1$, which is compatible with the most common assumptions of the MHD models \cite{biskamp1997nonlinear}, the resultant non-dimensionalized system will retain the same form as the dimensional one as follows:
 \begin{subequations}
\begin{align}
    \displaystyle &\partial_{t}f_s(\bs{x},\bs{v},t)+\bs{v}\cdot\nabla_{\bs{x}}f_s(\bs{x},\bs{v},t)-\bs{E}(\bs{x},t)\cdot\nabla_{\bs{v}}f_s(\bs{x},\bs{v},t)=0\\
    \displaystyle &\lambda^{2}\nabla\cdot\bs{E}(\bs{x},t)=\rho(\bs{x},t)\label{eq: dimensionless-mP-Poisson}\\
    \displaystyle &\bs{E}(\bs{x},t)=-\nabla\phi(\bs{x},t)\label{eq: dimensionless-mP-phi}
\end{align}
\label{eq: dimensionless-mP}%
\end{subequations}
where $\lambda=\lambda_D/x_0$. Considering our one-species assumption, the subscript $s$ of distribution function $f$ is omitted, and $\rho(\bs{x},t)=1-n(\bs{x},t)$. All the variables and parameters have been appropriately normalized based on Table \ref{tab: normalization} hereunder unless otherwise specified.

The main difficulty of developing energy-conserving schemes for system \eqref{eq: dimensionless-mP} resides in Gauss's law \eqref{eq: dimensionless-mP-Poisson}, in which the change of the electric field depends instantaneously on the density function. To overcome this obstacle, system \eqref{eq: dimensionless-mP} is reformulated into the following Vlasov-Amp\`ere (VA) system:
\begin{subequations}\label{eq: classical-VA}
\begin{align}
    \displaystyle &\partial_{t}f+\bs{v}\cdot\nabla_{\bs{x}}f-\bs{E}\cdot\nabla_{\bs{v}}f=0,\label{eq: classical-VA-Vlasov}\\
    \displaystyle &\lambda^{2}\partial_{t}\bs{E}-\bs{\Theta}+\bs{J}=\bs 0,\label{eq: classical-VA-Ampere}\\
    \displaystyle &\nabla\times\bs{E}=\bs{0},\label{eq: classical-VA-curlE}\\
    \displaystyle &\nabla\cdot\bs{\Theta}\label{eq: classical-VA-divG}=0,
\end{align}
\end{subequations}
where $\bs{\Theta}(\bs{x},t)$ is an artificial solenoidal field to guarantee the equivalence of VP and VA models, as is shown in Theorem \ref{thm:equivalence1}.

\begin{theorem}\label{thm:equivalence1}
The Vlasov-Poisson system \eqref{eq: dimensionless-mP} and the Vlasov-Amp\`ere system \eqref{eq: classical-VA} are equivalent, given that $\lambda^{2}\nabla\cdot\bs{E}(\bs{x})=\rho(\bs{x})$ is satisfied at $t=0$.
\end{theorem}

\begin{proof}
(VP $\Rightarrow VA$):   Let us first derive the proposed Vlasov-Amp\`ere equations from the Vlasov-Poisson equations. One takes the temporal derivative of Eq. \eqref{eq: dimensionless-mP-Poisson} to obtain
\begin{equation}
    \displaystyle \lambda^{2}\nabla\cdot\partial_{t}\bs{E}=\partial_{t}\rho.
\end{equation}
The charge continuity equation \eqref{eq: chargeeq} is then substituted into the above, which leads to
\begin{equation}
\nabla\cdot\Big( \lambda^{2}\partial_{t}\bs{E}+\bs{J} \Big)=0.
\end{equation}
This means $\epsilon_{0}\lambda^2\partial_{t}\bs{E}+\bs{J}$ is a solenoidal field, which can be represented by an auxiliary divergence-free field $\bs{\Theta}$, i.e.,
\begin{equation}\label{eq: newAmp}
    \displaystyle \lambda^{2}\partial_{t}\bs{E}+\bs{J}=\bs{\Theta}.
\end{equation}
Due to Eq. \eqref{eq: dimensionless-mP-phi} and the fact that $\bs \Theta$ is solenoidal, one arrives at the VA system \eqref{eq: classical-VA}.\\
~\\
\noindent(VP $\Leftarrow$ VA): Now it remains to recover the VP system \eqref{eq: dimensionless-mP} from the VA system \eqref{eq: classical-VA}. Firstly, Eq. \eqref{eq: dimensionless-mP-phi} can be directly obtained by the irrotational property of $\bs E$ in Eq. \eqref{eq: classical-VA-curlE}. Then one takes the divergence of Eq. \eqref{eq: classical-VA-Ampere} to obtain
\begin{equation}
    \displaystyle \lambda^{2} \nabla\cdot\partial_{t}\bs{E}+\nabla\cdot\bs{J}=0,
\end{equation}
where the divergence-free constraint of $\bs{\Theta}$ in Eq. \eqref{eq: classical-VA-divG} has been used. Again inserting the charge continuity equation \eqref{eq: chargeeq} into the above, one arrives at
\begin{equation}
\partial_{t}( \nabla\cdot \lambda^{2} \bs{E}-\rho)=0,
\end{equation}
which indicates that the Gauss's law \eqref{eq: dimensionless-mP-Poisson} is guaranteed if it is obeyed at $t=0$.
\end{proof}

\begin{remark} The idea of developing energy-conserving scheme via the reformulation of the VP system to the VA system can be traced back to Chen \emph{et al.} \cite{chen2011energy}, where they proposed the following one-dimensional VA reformulation:
\begin{subequations}
\begin{align}
  \displaystyle &\partial_{t}f+v \partial_x f-E \partial_v f=0,\label{eq: 1d-VA-Vlasov}\\
 &\lambda^{2}\frac{\partial E}{\partial t}+J=\langle J\rangle,\label{eq: 1d-Amp}
\end{align}
\end{subequations}
and $\langle J\rangle=\int Jdx/\int dx$ is proven to be a constant and independent of space $x$  and time $t$. However, for a system with two/three spatial dimensions, $\langle J\rangle$ in Eq. \eqref{eq: 1d-Amp} will be replaced by a solenoidal field depending on $\bs x$ and $t$ as shown in Eq. \eqref{eq: newAmp} and numerical difficulties are induced in satisfying the curl-and divergence-free constraints in Eqs. \eqref{eq: classical-VA-curlE}-\eqref{eq: classical-VA-divG}. It is worthwhile to point out that the reformulation of the Poisson equation into the curl-free constrained Amp\`ere equation has been explored in \cite{qiao2022amp} for the Poisson-Nernst-Planck system, where a local curl-free relation iterative algorithm~(originated in \cite{maggs2002local} for Coulomb interactions) is adopted to deal with these constraints, and it was further shown in \cite{qiao2022scheme} that structure-preserving schemes can be constructed based on the Amp\`ere formulation.
\end{remark}

The energy conservation law of the Vlasov-Amp\`ere system \eqref{eq: classical-VA} is expounded in Theorem \ref{thm: VA-conservation-law}.
\begin{theorem}[Energy conservation law of the Vlasov-Amp\`ere system]
The Vlasov-Amp\`ere system and its equivalent Vlasov-Poisson system with the periodic boundary condition satisfy the following energy conservation law:
\begin{equation}\label{eq: ecva}
        \displaystyle \frac{d}{dt}\Big( \frac{\lambda^{2}}{2}\int_{\Omega_{\bs x}}|\bs{E}|^{2}d\bs{x}+\sum\limits_{s}\frac{1}{2}\iint_{\Omega_{\bs x}\times \Omega_{\bs v}} f|\bs{v}|^{2}d\bs{x}d\bs{v}\Big)= 0.
    \end{equation}
    \label{thm: VA-conservation-law}
\end{theorem}

\begin{proof}
By multiplying Eq. \eqref{eq: classical-VA-Vlasov}  by $|\bs{v}|^{2}/2$ and integrating the resultant equation over $\Omega_{\bs{x}}$, $\Omega_{\bs{v}}$, we have
\begin{equation}
    \displaystyle\frac{1}{2} \iint_{\Omega_{\bs x}\times \Omega_{\bs v}}\left[f|\bs{v}|^{2}+(\bs{v}\cdot\nabla_{\bs{x}}f)|\bs{v}|^{2}-(\bs{E}\cdot\nabla_{\bs{v}}f)|\bs{v}|^{2}    \right]d\bs{x}d\bs{v}=0.
    \label{eq: proof-exact-conservation-multiply-v}
\end{equation}
Then through taking integration by parts and imposing periodic or Dirichlet boundary conditions, the second term in the above equation vanishes, and the last term becomes
\begin{equation}\label{eq: 3terms}
 \frac{1}{2} \iint_{\Omega_{\bs x}\times \Omega_{\bs v}}  (\bs{E}\cdot\nabla_{\bs{v}}f)|\bs{v}|^{2} d\bs{x}d\bs{v}=-\int_{\Omega_{\bs x}} \bs E\cdot \Big(\int_{\Omega_{\bs v}} f \bs v d\bs v \Big) d\bs x=\int_{\Omega_{\bs x}}  \bs E\cdot \bs J d\bs x.
\end{equation}
We multiply $\bs E$ with Eq. \eqref{eq: classical-VA-Ampere} and integrate the resultant equation over $\Omega_{\bs{x}}$ to obtain
\begin{equation}\label{eq: E2}
\frac{d}{dt} \int_{\Omega_{\bs x}}  \frac{\lambda^{2}}{2} | \bs E|^2 d\bs x-\int_{\Omega_{\bs x}} \bs \Theta\cdot \bs Ed\bs x+\int_{\Omega_{\bs x}}  \bs E\cdot \bs J d\bs x=0.
\end{equation}
By the fact that $\nabla \times \bs E=0$ in Eq. \eqref{eq: classical-VA-curlE}, we can rewrite $\bs E$ by $\bs E=-\nabla \phi$ and insert it into the second term of the above equation, during which we find
\begin{equation}\label{eq: GEvanish}
\int_{\Omega_{\bs x}} \bs \Theta\cdot \bs Ed \bs x=0,
\end{equation}
where the integration by parts, the divergence-free property of $\bs \Theta$ in Eq \eqref{eq: classical-VA-divG} and the boundary conditions have been used. Thus, the sum of these resultant Eqs. \eqref{eq: 3terms}, \eqref{eq: E2} and the definition of $\bs J$ in Eq. \eqref{eq: rhoj} directly give rise to the desired energy conservation law in Eq. \eqref{eq: ecva}.
\end{proof}

\section{A novel family of Fourier basis preserving the curl-/divergence-free constraints}
\label{sec: div-curl-free-expansion}

In this section, we propose a systematic way to construct a Fourier approximation basis preserving the divergence-free or curl-free constraints point-wisely. Without loss of generality, we consider the VA system in $\Omega_{\bs x}=[0,L]^d,\; d=2\ \text{and}\ 3$ with periodic boundary condition. Let us denote the scalar Fourier basis function by
\begin{equation}\label{eq: fep}
E_{\bs p}(\bs x)=e^{\frac{2\pi \ri }{L}\bs p \cdot \bs x}, \quad \bs p\in \mathbb{N}^d, \quad \mathbb{N}=[-M/2+1,\cdots,0,1, \cdots M/2]^d,\;\;d=2,3,
\end{equation}
and $\{\bs{e}_{i}\}_{i=1}^d$ the canonical basis vectors along each coordinate axis.
The construction of desired structure-preserving vectorial Fourier basis functions relies on the derivative relation and orthogonal property of $E_{\bs p}(\bs x)$:
\begin{itemize}
\item[1]  \underline{The derivative relation}:
\begin{equation}\label{eq: derivp}
 \frac{\partial E_{\bs p}(\bs x)}{\partial x_j}=\frac{2\pi \ri }{L} p_j  E_{\bs p}(\bs x),
\end{equation}
\item[2] \underline{The orthogonal property}:
\begin{equation}\label{eq: orth}
\frac{1}{L^d}\int_{[0,L]^d}E_{\bs p}(\bs x) \bar E_{\bs q}(\bs x)d\bs x=\delta_{\bs p,\bs q},\;\; \bs p,\bs q\in  \mathbb{N}^d,
\end{equation}
where $\delta_{\bs p,\bs q}$ is the Dirac delta function.
\end{itemize}

Under 2D Cartesian coordinates and given an arbitrary vector field $\bs v(\bs x)=(v_1,v_2)^T$, the curl and divergence operators respectively take the form
\begin{equation}\label{eq: curldiv}
\nabla \times \bs v=\frac{\partial v_2}{\partial x_1}- \frac{\partial v_1}{\partial x_2},\quad \nabla \cdot \bs v=\frac{\partial v_1}{\partial x_1}+ \frac{\partial v_2}{\partial x_2}.
\end{equation}
For any square-integral periodic function $\bs v$, it can be wonderfully approximated by
\begin{equation}\label{eq: der1}
\bs v_{N}=\sum_{(m,n)\in \mathbb{N}^2} (v_{mn}^1 \bs e_1+v_{mn}^2 \bs e_2)E_{(m,n)}(\bs x).
\end{equation}
By taking the divergence operator for $\bs v_{N}$ and presuming that $\nabla \cdot \bs v_{N}=0$, one directly obtains from the derivative relation \eqref{eq: derivp} that
\begin{equation}\label{eq: der2}
\nabla \cdot \bs v_{N}=\frac{2\pi \ri }{L} \sum_{(m,n)\in \mathbb{N}^2} (m v_{mn}^1 +n v_{mn}^2 )E_{(m,n)}(\bs x)=0,
\end{equation}
which together with the orthogonal property \eqref{eq: orth} of $\{ E_{(m,n)}(\bs x) \}$ under $L^2$ inner product leads to $m v_{mn}^1 +n v_{mn}^2=0,\;\; \forall m,n\in \mathbb{N}$. This means that $\{v_{mn}^1, v_{mn}^2 \}$ can be represented by only one free variable except in the case when $m=n=0$:
\begin{equation}\label{eq: der4}
v_{mn}^1= n \hat v_{mn} , \quad v_{mn}^2= -m \hat v_{mn} , \;\; (m,n)\in \mathbb{N}^2_{+}= \mathbb{N}^2/(0,0).
\end{equation}
Similarly, by taking the curl operator and using the derivative relation and orthogonal property in Eqs. \eqref{eq: derivp}-\eqref{eq: orth}, the expansion coefficients of the curl-free vector function $\bs v_N$ satisfy $mv_{mn}^2-nv_{mn}^1  =0,\;\;\forall m,n\in \mathbb{N}$, thus $\{v_{mn}^1, v_{mn}^2 \}$, in this case, can be represented by
\begin{equation}\label{eq: der6}
v_{mn}^1= m \tilde v_{mn} , \quad v_{mn}^2= n \tilde v_{mn} , \;\; (m,n)\in \mathbb{N}^2_{+}= \mathbb{N}^2/(0,0).
\end{equation}
Consequently, we summarize the 2D divergence-free and curl-free Fourier basis in Proposition \ref{prop: p1}.

\begin{proposition} \label{prop: p1}
Define
\begin{equation}\label{eq: dc}
\bs d_{mn}=(n,-m)^T, \quad \bs c_{mn}=(m,n)^T, \quad (m,n)\in \mathbb{N}^2_{+},
\end{equation}
the 2D divergence-free Fourier basis takes the form
\begin{equation}
\mathbb{D}^{2}_N={\rm span}\left\{ \big\{ \bs d_{mn} E_{(m,n)}(\bs x)   \big \}_{(m,n)\in \mathbb{N}^2_{+}}, \bs e_1,\bs e_2 \right\},
\end{equation}
while the 2D curl-free Fourier basis reads
\begin{equation}\label{eq: C2}
\mathbb{C}^{2}_N={\rm span}\left\{ \big\{ \bs c_{mn} E_{(m,n)}(\bs x)   \big \}_{(m,n)\in \mathbb{N}^2_{+}}, \bs e_1,\bs e_2 \right\}.
\end{equation}
\end{proposition}

In 3D Cartesian coordinates, the curl and divergence operators take the form
\begin{equation}\label{eq: curl3d}
\nabla \times \bs v=\Big( \frac{\partial v_3}{\partial x_2}- \frac{\partial v_2}{\partial x_3}, \; \frac{\partial v_1}{\partial x_3}-\frac{\partial v_3}{\partial x_1},\; \frac{\partial v_2}{\partial x_1}- \frac{\partial v_1}{\partial x_2} \Big)^T,\quad \nabla \cdot \bs v=\frac{\partial v_1}{\partial x_1}+\frac{\partial v_2}{\partial x_2}+\frac{\partial v_3}{\partial x_3},
\end{equation}
for any $\bs v(\bs x)=(v_1,v_2,v_3)^T$. Similarly, for any square-integral periodic function $\bs v$ in $\Omega_{\bs x}=[0,L]^3$, it can be expanded by
\begin{equation}\label{eq: 3dexp}
\bs v_{N}=\sum_{(m,n,l)\in \mathbb{N}^3} (v_{mn}^1 \bs e_1+v_{mn}^2 \bs e_2+v_{mn}^3 \bs e_3)E_{(m,n,l)}(\bs x).
\end{equation}
With the help of the expression \eqref{eq: 3dexp}, the derivative relation \eqref{eq: derivp}, the orthogonal property \eqref{eq: orth} and the fact that the vector function $\bs v_{N}$ is divergence-free or curl-free, one can derive the corresponding Fourier bases for three dimensions, which are summarized in Proposition \ref{prop: p2}.

\begin{proposition}\label{prop: p2}
Define two sets of tensors $\big \{ \bs d^1_{mnl}, \bs d^2_{mnl}\big\}_{(m,n,l)\in \mathbb{N}^3_{+}}$ for  $\mathbb{N}^3_{+}=\mathbb{N}^3/ (0,0,0)^T$,
\begin{equation}
\begin{aligned}\label{eq: 3dzeta}
& \bs d^1_{mnl}=(n,-m,0)^T,\quad  \bs d^2_{mnl}=(l,0,-m)^T, \qquad m\neq0,\\
& \bs d^1_{0nl}=\bs e_1, \quad  \bs d^2_{0nl}=(0,-l,n), \qquad m=0,\;\; n\neq0,\\
&  \bs d^1_{00l}=\bs e_1, \quad \bs d^2_{00l}=\bs e_2,\qquad m=n=0, \;l\neq 0,
\end{aligned}
\end{equation}
3D divergence-free Fourier basis takes the form
\begin{equation}
\mathbb{D}^{3}_N={\rm span}\big( \big\{ \bs d^1_{mnl} E_{(m,n,l)}(\bs x) ,\,\bs d^2_{mnl} E_{(m,n,l)}(\bs x)  \big \}_{(m,n,l)\in \mathbb{N}^3_{+}}, \bs e_1,\bs e_2,\bs e_3 \big).
\end{equation}
While for tensors $\big \{ \bs c_{mnl}=(m,n,l)^T\big\}_{(m,n,l)\in \mathbb{N}^3_{+}}$, 3D curl-free Fourier basis reads
\begin{equation}\label{eq: C3}
\mathbb{C}^{3}_N={\rm span}\big( \big\{ \bs c_{mnl} E_{(m,n,l)}(\bs x)   \big \}_{(m,n,l)\in \mathbb{N}^3_{+}}, \bs e_1,\bs e_2, \bs e_3 \big).
\end{equation}
\end{proposition}

\begin{remark}
The proposed structure-preserving Fourier method serves as a proper way for discretizing the solenoidal field $\bs \Theta$
and the irrotational field $\bs E$. In what follows, one observes that under a suitable Galerkin formulation, the discretization of  $\bs \Theta$ is not necessary, and only the curl-free basis is utilized. Nevertheless, the divergence-free discretization developed here will play a significant role for the solution of Vlasov-Maxwell system, especially for the exact preservation of the magnetic Gauss's law.
\end{remark}

\section{An energy-conserving PIC method with asymptotic-preserving preconditioner}
\label{sec: fully-implicit}

\subsection{Particle-in-cell discretization for Vlasov equation}

The PIC method discretizes the Vlasov equation by a sequence of macro particles, with the advantage of reducing the solution of the Vlasov equation in the six-dimensional phase space into solving Newton's second law of macro particles. To be more specific, the distribution function $f$ is approximated by
\begin{equation}
    \displaystyle f(\bs{x},\bs{v},t)=\sum\limits_{p=1}\limits^{N}w_{p}S(\bs{x}-\bs{x}_{p})\delta(\bs{v}-\bs{v}_{p}),
    \label{eq: dsitribution-f-pic}
\end{equation}
where $\bs{x}_{p}$, $\bs{v}_{p}$ and $w_{p}$ are the position, velocity and weight of the macro particle $p$, individually, and $N$ is the total number of macro particles. $\delta(\bs v)$ is the Dirac delta function, and $S(\bs{x})$ is the shape function of the macro particle, chosen as a particular compactly-supported symmetric function with unit integral $\int_{\Omega_{x}}S(\bs{x})d\bs{x}=1$. There are several standard options for the shape function, such as B-spline basis function \cite{2017SMILEI}, cosine function with a cut-off radius, Gaussian function and polynomials with unit integral (see, e.g. \cite{2006High}).

By substituting Eq. \eqref{eq: dsitribution-f-pic} into the Vlasov equation \eqref{eq: classical-VA-Vlasov}, and taking the first-order momentum of the Vlasov equation with respect to $\bs x$ and $\bs v$ over $\Omega_{\bs{x}}$ and $\Omega_{\bs{v}}$, respectively, together with the properties of shape function $S(\bs x)$, one arrives at a sequence of the particle motion equations for $p=1,\ldots,N$:
\begin{equation}\label{eq: motioneq}
\frac{d\bs{x}_{p}}{dt}=\bs{v}_{p},\quad \frac{d\bs{v}_{p}}{dt}=-\bs{E}_{p}, \;\;{\rm with}\;\; \bs{E}_{p}=\int_{\Omega_{\bs{x}}}\bs{E}(\bs{x})S(\bs{x}-\bs{x}_{p})d\bs{x}.
\end{equation}

\subsection{Strong and weak formulation of the particle-Amp\`ere system}
Consequently, the resultant particle-Amp\`ere system consisting of Eqs. \eqref{eq: motioneq}, \eqref{eq: classical-VA-Ampere}-\eqref{eq: classical-VA-divG} are summarized as follows:
\begin{subequations}\label{eq: pva}
\begin{align}
    \displaystyle &\frac{d\bs{x}_{p}}{dt}=\bs{v}_{p},\label{eq: pnewton1}\\
    \displaystyle &\frac{d\bs{v}_{p}}{dt}=-\bs{E}_{p},\label{eq: pnewton2}\\
    \displaystyle &\lambda^{2}\partial_{t}\bs{E}(\bs{x},t)-\bs{\Theta}(\bs{x},t)+\bs{J}(\bs{x},\bs{v},t)=\bs 0, \label{eq:pE1}\\
    \displaystyle &\nabla\times\bs{E}(\bs{x},t)=\bs{0},\label{eq: pE2}\\
    \displaystyle &\nabla\cdot\bs{\Theta}(\bs{x},t)=0,\label{eq: pE3}
\end{align}
\label{eq: dimensionless-motion-Ampere}%
\end{subequations}
where the current density $\bs{J}$ is consistently discretized as
\begin{equation}\label{eq: sourcenj}
    \displaystyle \bs{J}(\bs{x},\bs{v},t)=-\sum\limits_{p=1}\limits^{N_{s}}w_{p}S(\bs{x}-\bs{x}_{p})\bs{v}_{p},
\end{equation}
and in similar manner, the number density $n$, charge density $\rho$ are approximated by macro particles as
\begin{equation}\label{eq: sourcenj}
    \displaystyle n(\bs{x},t)=\sum\limits_{p=1}\limits^{N_{s}}w_{p}S(\bs{x}-\bs{x}_{p}),\quad\rho(\bs{x},t)=1-n(\bs{x},t).
\end{equation}

\begin{remark}
Without external sources, the particle-Amp\`ere system \eqref{eq: pva} satisfies the energy conservation law, i.e., the total energy
\begin{equation}
     \displaystyle E_{total}(t)=\frac{\lambda^{2}}{2}\int_{\Omega_{\bs{x}}}|\bs{E}|^{2}d\bs{x}+\sum\limits_{p=1}\limits^{N_{s}}\frac{w_{p}}{2}\bs{v}_{p}^{2}
 \end{equation}
 remains constant with time. And it is crucial for numerical schemes to preserve the energy conservation law for accurate and robust long-time simulations, especially when there exist multiple physical scales and large time steps are preferred.
\end{remark}

In order to propose our numerical scheme, we briefly introduce some basic notations of Sobolev spaces. Let $L^2_{\rm per}(\Omega_{\bs x})$ be the Hilbert space of square-integrable periodic functions on $\Omega_{\bs x}$ with norm $\| \cdot \|_{\Omega_{\bs x}}$ and inner product $(\cdot, \cdot)_{\Omega_{\bs x}}$. We introduce
\begin{equation}\label{eq: Hcurldiv}
\begin{aligned}
& \bs H^{\rm curl}(\Omega_{\bs x})=\big\{\bs v\in (L^2_{\rm per}(\Omega_{\bs x}))^d, \nabla \times \bs v\in (L^2_{\rm per}(\Omega_{\bs x}))^d\big\},\\
&  \bs H^{\rm div}(\Omega_{\bs x})=\big\{\bs v\in (L^2_{\rm per}(\Omega_{\bs x}))^d, \nabla \cdot \bs v\in L^2_{\rm per}(\Omega_{\bs x})\big\},
\end{aligned}
\end{equation}
and denote
\begin{equation}
\begin{aligned}
& \bs H^{\rm curl,0}(\Omega_{\bs x}):=\{\bs v\in \bs H^{\rm curl}(\Omega_{\bs x}), \;\nabla \times \bs v=\bs 0  \} ,\\
&  \bs H^{\rm div,0}(\Omega_{\bs x}):=\{\bs v\in \bs H^{\rm div}(\Omega_{\bs x}), \;\nabla \cdot \bs v= 0  \}.
\end{aligned}
\end{equation}
We henceforward omit $\Omega_{\bs x}$ in the notations of the inner product and function spaces if there is no ambiguity.

The weak formulation of the particle-Amp\`ere system \eqref{eq: pva}  is summarized as follows: 
Find $\{\bs x_p, \;\bs v_p \}_{p=1}^{N}$ and $\bs E\in \bs H^{\rm curl,0}$, such that
\begin{subequations}\label{eq: pwva}
\begin{align}
    \displaystyle &\frac{d\bs{x}_{p}}{dt}=\bs{v}_{p},\label{eq: pweak1}\\
   \displaystyle &\frac{d\bs{v}_{p}}{dt}=-\int_{\Omega_{\bs{x}}}\bs{E}(\bs{x})S(\bs{x}-\bs{x}_{p})d\bs{x},\label{eq: pweak2}\\
    \displaystyle &\lambda^{2}( \partial_{t}\bs{E}, \bs \varphi)+(\bs{J}, \bs \varphi)=\bs 0,\;\; {\forall \bs \varphi \in \bs H^{\rm curl,0}}. \label{eq: pweak3}
\end{align}
\end{subequations}
\begin{remark}
Note that the term associated with $\bs{\Theta}$ vanishes from Eq. \eqref{eq: pweak3} as
\begin{equation*}
(\bs{\Theta}, \bs \varphi)=0,\;\;\forall\, \bs{\Theta}\in  \bs H^{\rm div,0}, \; \bs \varphi \in \bs H^{\rm curl,0}.
\end{equation*}
The proof of this equality  resembles the derivation for Eq. \eqref{eq: GEvanish}.
\end{remark}

In the following two sections, we propose two energy-conserving schemes for the particle-Amp\`ere system \eqref{eq: pwva}. One is based on the CN temporal discretization, where an asymptotic-preserving preconditioner is developed to accelerate the convergence speed. The other one is based on the Strang operator-splitting method. For both schemes, we emphasize the necessity of using the proposed spatial discretizations with exact curl-free constraint in order to preserve the energy conservation law.

\subsection{The fully-implicit energy-conserving scheme}
The following CN scheme can achieve the discrete energy conservation law: Given $\{\bs x_p^n, \;\bs v_p^n \}_{p=1}^{N}$ and $\bs E^n\in \bs H^{\rm curl,0}$, find $\{\bs x_p^{n+1}, \;\bs v_p^{n+1} \}_{p=1}^{N}$ $\in \mathbb{R}^d$ and $\bs E^{n+1/2}\in \mathbb{C}_N^d$ defined in Propositions \ref{prop: p1}-\ref{prop: p2}, such that
\begin{subequations}
\begin{align}
    \displaystyle &\frac{\bs{x}_{p}^{n+1}-\bs{x}_{p}^{n}}{\Delta t}=\bs{v}_{p}^{n+1/2}, \label{eq: schemeeq1}\\
    \displaystyle &\frac{\bs{v}_{p}^{n+1}-\bs{v}_{p}^{n}}{\Delta t}=-\int_{\Omega_{\bs{x}}}\bs{E}^{n+1/2}S(\bs{x}-\bs{x}_{p}^{n+1/2})d\bs{x},\label{eq: schemeeq2}\\
& \frac{2\lambda^{2}}{\Delta t}( \bs{E}^{n+1/2}-\bs E^n, \bs \varphi)+(\bs{J}^{n+1/2}, \bs \varphi)=\bs 0,\;\; \forall \bs \varphi \in \mathbb{C}_N^d, \label{eq: disEform}
\end{align}
    \label{eq: discrete-dimensionless-motion-Ampere}%
\end{subequations}
where variables discretized at the half-time steps are given by
\begin{equation}\label{eq:jn12}
    \begin{aligned}
    \displaystyle &\bs{x}_{p}^{n+1/2}=\frac{\bs{x}_{p}^{n+1}+\bs{x}_{p}^{n}}{2},\quad\bs{v}_{p}^{n+1/2}=\frac{\bs{v}_{p}^{n+1}+\bs{v}_{p}^{n}}{2},\quad\bs{E}^{n+1/2}=\frac{\bs{E}^{n+1}+\bs{E}^{n}}{2},\\
    &\bs{J}^{n+1/2}=\bs{J}(\bs{x}_{p}^{n+1/2},\bs{v}_{p}^{n+1/2})=-\sum\limits_{p=1}\limits^{N_{s}}w_{p}S(\bs{x}-\bs{x}_{p}^{n+1/2})\bs{v}_{p}^{n+1/2}.
    \end{aligned}
\end{equation}

\begin{theorem}\label{thm: thmec}
The fully-implicit CN scheme \eqref{eq: discrete-dimensionless-motion-Ampere} satisfies the discrete energy conservation law, namely,
\begin{equation}\label{eq: paconlaw}
\frac{1}{2}\lambda^{2}\int_{\Omega_{\bs{x}}}|\bs{E}^{n+1}|^{2}d\bs{x}+\sum\limits_{p=1}\limits^{N_{s}}\frac{1}{2}w_{p}|\bs{v}_{p}^{n+1}|^{2}
   = \frac{1}{2}\lambda^{2}\int_{\Omega_{\bs{x}}}|\bs{E}^n|^{2}d\bs{x}+\sum\limits_{p=1}\limits^{N_{s}}\frac{1}{2}w_{p}|\bs{v}_{p}^n|^{2}.
 \end{equation}
 \end{theorem}

 \begin{proof}
 Let $\bs \varphi=\bs{E}^{n+1/2}$ in Eq. \eqref{eq: disEform}, we have
 \begin{equation}
    \frac{\lambda^{2}}{2} \left(\|\bs{E}^{n+1}\|^{2}-\|\bs{E}^{n}\|^{2}\right)+\int_{\Omega_{\bs{x}}}\bs{J}^{n+1/2}\cdot\bs{E}^{n+1/2}d\bs{x}=\bs 0.
    \label{eq: energy-conservation-exact-proof-1}
\end{equation}
By the definitions of $\bs J^{n+1/2}$ in Eq. \eqref{eq:jn12} and the particle motion equations \eqref{eq: schemeeq1}-\eqref{eq: schemeeq2}, the last term in Eq. \eqref{eq: energy-conservation-exact-proof-1} can be simplified into
\begin{equation}
\begin{aligned}
    \displaystyle \int_{\Omega_{\bs{x}}}\bs{J}^{n+1/2}\cdot\bs{E}^{n+1/2}d\bs{x}&=-\sum\limits_{p=1}\limits^{N_{s}}w_{p}\int_{\Omega_{\bs{x}}} \bs{E}^{n+1/2}(\bs{x})S(\bs{x}-\bs{x}_{p}^{n+1/2})d\bs{x}\,\bs{v}_{p}^{n+1/2}\\
    &=\sum\limits_{p=1}\limits^{N_{s}}w_{p}\frac{\bs{v}_{p}^{n+1}-\bs{v}_{p}^{n}}{\Delta t}\bs{v}_{p}^{n+1/2}=\frac{w_{p}}{2}\sum\limits_{p=1}\limits^{N_{s}}\left[\left(\bs{v}_{p}^{n+1}\right)^{2}-\left(\bs{v}_{p}^{n}\right)^{2}\right],
    \end{aligned}
\end{equation}
which together with Eq. \eqref{eq: energy-conservation-exact-proof-1} lead to the desired result \eqref{eq: paconlaw}.
 \end{proof}

 \begin{remark}
It can be observed clearly that the exact preservation of the constraint $\nabla \times \bs{E}=\bs 0$ at the discrete level is indispensable for achieving discrete energy law.
 \end{remark}

 Though the proposed method guarantees energy conservation, thus suitable for long-time simulations, its efficiency is greatly restricted by the sizeable nonlinear system to be solved. For instance, for a system discretized with $N$  macro particles in $d\times d_1$ phase space for the Vlasov equation using the PIC method and $M$ Fourier modes for each spatial coordinate of the Amp\`ere equation, the number of unknowns in the nonlinear system is $O(M^d+dd_1 N)$, which is extremely large due to the necessity to adopt large $N$ values  ($N\geq 10^5$) to reduce stochastic noises.

 One way to overcome this difficulty is to employ the particle enslavement technique proposed in \cite{chen2011energy}, which regards $\bs E^{n+1/2}$ as the only primary unknown and $\{\bs x_p^{n+1}, \bs v_p^{n+1} \}_{p=1}^{N}$ as the intermediate variables. Specifically, once $\bs E^{n+1/2}$ is given, $\{\bs x_p^{n+1}, \bs v_p^{n+1} \}_{p=1}^{N}$ can be uniquely determined by Eqs. \eqref{eq: schemeeq1}-\eqref{eq: schemeeq2}, thus $\bs J^{n+1/2}$ can be calculated by Eq. \eqref{eq:jn12}. This
suggests that $\bs J^{n+1/2}$ can be regarded as a function of $\bs E^{n+1/2}$, i.e., $\bs J^{n+1/2}=\bs J^{n+1/2}(\bs E^{n+1/2})$. Consequently, the dimensions of the nonlinear system are successfully reduced to $O(M^d)$.

To fix the idea, let us consider the case with two dimensions. Rewrite Eq. \eqref{eq: disEform} into the following equivalent form
\begin{equation}\label{eq: Es}
{2\lambda^{2}}( \bs{E}^{n+1/2}, \bs \varphi)+(\Delta t \bs{J}^{n+1/2} -2\lambda^2 \bs E^n, \bs \varphi)=\bs 0,\;\; \forall \bs \varphi \in \mathbb{C}_N^d,
\end{equation}
and denote $\Delta t \bs{J}^{n+1/2} -2\lambda^2 \bs E^n:=\bs f^n.$ Let us expand $\bs f$ by
\begin{equation}\label{eq: fexp}
\bs f=\sum_{(m,n)\in \mathbb{N}^2} (f_{mn}^1\bs e_1+f_{mn}^2 \bs e_2) E_{(m,n)}(\bs x),
\end{equation}
and expand $\bs E^{n+1/2}$ using the curl-free Fourier basis in Eq. \eqref{eq: C2},
\begin{equation}\label{eq: Eexp}
\bs E^{n+1/2}(\bs x)=\xi_{00}^1\bs e_1 E_{(0,0)}(\bs x)  +\xi_{00}^2 \bs e_2 E_{(0,0)}(\bs x) +\sum_{(m,n)\in \mathbb{N}^2_{+}} \xi_{mn} \bs c_{mn} E_{(m,n)}(\bs x),
\end{equation}
with the unknown coefficients $\{\xi_{mn}, \xi_{00}^1, \xi_{00}^2 \}$  reordered into a column vector $\bs \xi$.

Inserting the expansions \eqref{eq: fexp} and \eqref{eq: Eexp} into Eq. \eqref{eq: Es} and taking $\bs \varphi$ by the curl-free basis functions in Eq. \eqref{eq: C2} lead to the following system in terms of $\bs \xi$
\begin{subequations}\label{eq: nonlinearsystem}
\begin{align}
& 2\lambda^2 \xi_{00}^1+f_{00}^1(\bs \xi)=0,\\
& 2\lambda^2 \xi_{00}^2+f_{00}^2(\bs \xi)=0,\\
& 2\lambda^2 (m^2+n^2)\xi_{mn}+( mf_{mn}^1(\bs \xi)+n f_{mn}^2(\bs \xi) )=0,\;\; (m,n)\in \mathbb{N}^2_{+}
\end{align}
\end{subequations}
which is denoted by $\bs {\mathcal F}(\bs \xi)=\bs 0$. Note that $\{ f_{mn}^i\}$ depend on $\bs \xi$ as $\bs J^{n+1/2}$ is uniquely determined by $\bs E^{n+1/2}$, thus the above system is a nonlinear one.

The numerical solution of such a strongly-coupled nonlinear system has proven to be challenging. Classical Newton-type methods for solving the nonlinear system require the computation of the Jacobian $\delta \bs {\mathcal F}(\bs \xi^{(k)})/\delta \bs \xi$ or Jacobian-vector multiplication $(\delta \bs {\mathcal F}(\bs \xi^{(k)})/\delta \bs \xi)\, \bs w$ for specific given vector $\bs w$, which is cumbersome to compute for the above nonlinear system \eqref{eq: nonlinearsystem}. Thus, we resort to Anderson-acceleration (AA) method \cite{2009Anderson,2022Superlinear,Anderson2018}, which is a derivative-free iteration method with improved convergence compared with the traditional fixed-point iteration method. We refer the readers to the Appendix for a summary of the AA algorithm.

Compared with the traditional Picard fixed-point iteration method, the AA algorithm does not require $\mathcal{T}$ to be a contractive operator. It has an improved convergence rate, thus can be very efficient for solving the nonlinear system \eqref{eq: nonlinearsystem}, when coupled with an effective preconditioner.

\begin{remark}
It is worthwhile to note that we omit the calculation of $\bs{J}^{n+1/2,(k)}$ in each AA iteration. Actually, once $\bs \xi^{(k)}$, i.e. $\bs{E}^{n+1/2,(k)}$ is given, Eqs. \eqref{eq: schemeeq1}-\eqref{eq: schemeeq2} can be reformulated to a nonlinear problem for ${\bs v_p^{n+1/2,(k)}}$ as
\begin{equation}
2(\bs v_p^{n+1/2,(k)} -\bs v_p^n)+\Delta t \int_{\Omega_{\bs{x}}}\bs{E}^{n+1/2,(k)}S\Big(\bs{x}- \bs x_p^n-\frac{\Delta t}{2}\bs v_p^{n+1/2,(k)}   \Big )d\bs{x}=0,
\end{equation}
which can also be solved efficiently by the Anderson-acceleration algorithm. Once $\bs v_p^{n+1/2,(k)}$ is calculated, one can obtain ${\bs x_p^{n+1/2,(k)}}$ by $\bs x_p^{n+1/2,(k)}=\bs x_p^n+({\Delta t}/2) \bs v_p^{n+1/2,(k)}$
and compute $\bs{J}^{n+1/2,(k)}$  readily by Eq. \eqref{eq: sourcenj}.
\end{remark}

\subsection{Asymptotic-preserving preconditioner}
 \label{subsec: AP-preconditioner}

 Though the proposed Anderson-accelerated fully-implicit scheme \eqref{eq: discrete-dimensionless-motion-Ampere} with structure-preserving Fourier discretizations in Propositions \ref{prop: p1}-\ref{prop: p2} guarantees the energy conservation, it suffers from the difficulty of convergence for solving the resulting nonlinear system. When the system approaches the quasi-neutral limit, i.e. $\lambda \rightarrow 0$, nonlinear coupling between the Amp\`ere equation \eqref{eq:pE1} and the particle motion equations \eqref{eq: pnewton1}-\eqref{eq: pnewton2} are gradually magnified, and the nonlinear system becomes notoriously difficult to solve. We propose an effective preconditioner to accelerate the convergence, which can be viewed as a linearized approximation for the time-discretized Amp\`ere equation \eqref{eq: disEform}, for both moderate and small $\lambda$ values.

 We start from the generalized Ohm's law \cite{Somov2007,4ec3d3fc98db4e05b6bed56784f1d2c1,Kandus_2008}
\begin{equation}
    \displaystyle \frac{\partial\bs{J}(\bs x,\bs v,t)}{\partial t}=\nabla\cdot \mathcal{S}(\bs x,t)+n(\bs x,t)\bs{E}(\bs x,t),
    \label{eq: generalized-Ohm-law}
\end{equation}
which is obtained by taking the first-order momentum of the Vlasov equation with respect to the velocity field.  Here, $\mathcal{S}(\bs x,t)$ is the stress tensor defined by
\begin{equation}\label{eq: stress}
\mathcal{S}=\int_{\Omega_{\bs v}} f\bs{v} \otimes \bs vd\bs{v}
\end{equation}
We use the Lie-Trotter operator splitting technique \cite{WANG2022117,Lie-Trotter2021,doi:10.1137/140962644} to split Eq. \eqref{eq: generalized-Ohm-law} into the following two subproblems:
     \begin{subequations}\label{eq: ohmdis}
     \begin{align}
 &    \partial_t \bs J=\nabla \cdot \mathcal{S}, \label{eq: js1}\\
 &  \partial_t \bs J=n\bs E.\label{eq: implJ}
 \end{align}%
\end{subequations}
By the definitions of $\bs J$ and $\mathcal{S}$, using macro particles to approximate $f(\bs x,\bs v,t)$,
and taking integral of Eq. \eqref{eq: js1} over the spatial domain,
we arrive at a sequence of ODEs for the macro particles
 \begin{equation} \label{eq:xpstar}
\displaystyle \frac{d \bs x_p}{dt}=\bs v_p,\quad \displaystyle \frac{d \bs v_p}{dt}=\bs 0,
 \end{equation}
 which can be solved analytically by
 \begin{equation}
 \bs v_p^{*,n+1/2}= \bs v_p^{n}, \quad \bs x_{p}^{*,n+1/2}=\bs x_p^n +\frac{1}{2}\Delta t \bs v_p^n.
 \end{equation}
 Then, by the Lie-Trotter splitting scheme, subproblem \eqref{eq: implJ} is further discretized by
 \begin{equation}
    \displaystyle \frac{\bs{J}^{n+1/2}-\bs{J}^{\ast,n+1/2}}{\Delta t /2}=n^{*,n+1/2}\bs{E}^{n+1/2},
    \label{eq: new-preconditioner-approximation-J}
\end{equation}
where $n^{*,n+1/2}=n(\bs{x}_{p}^{*,n+1/2})$ and $\bs{J}^{\ast,n+1/2}=\bs{J}(\bs{x}_{p}^{*,n+1/2},\bs{v}_{p}^{*,n+1/2}).$  Consequently, $\bs{J}^{n+1/2}$ can be approximated by
\begin{equation}\label{eq: Japprox}
\bs{J}^{n+1/2}=\bs{J}^{\ast,n+1/2}+\frac{\Delta t}{2} n^{*,n+1/2}\bs{E}^{n+1/2}.
\end{equation}
Inserting the above approximation into Eq. \eqref{eq: disEform}, we have
\begin{equation}\label{eq: preconditioner}
\Big( \big(2\lambda^2+({\Delta t^2 }/{2})n^{*,n+1/2} \big) \bs E^{n+1/2}, \bs \varphi \Big)=(2\lambda^2 \bs E^n-\Delta t \bs J^{*,n+1/2}, \bs \varphi ).
\end{equation}
which is a uniform approximation for the discretized Amp\`ere equation \eqref{eq: disEform} for $\lambda$.
We substitute the expansions \eqref{eq: Eexp} of $\bs E^{n+1/2}$ into Eq. \eqref{eq: preconditioner} and take $\bs \varphi$ by the curl-free basis functions in Eq. \eqref{eq: C2}, then a linear system for the expansion coefficients $\bs \xi$ is given as 
\begin{equation}
\mathcal{M}\bs \xi=\bs b.
\end{equation}
Consequently, $\mathcal{M}$ can be used as a preconditioner for solving the nonlinear system \eqref{eq: nonlinearsystem}, for both moderate and small $\lambda$ values, thus making for an asymptotic-preserving preconditioner.

\begin{remark}\label{rem: Dprecond}
The preconditioner $\mathcal{M}$ itself can be inverted easily by a few iterations using GMRES preconditioned by diagonal matrix $\mathcal{D}$, which is obtained by replacing $\bs E^{n+1/2}$ and  $\bs \varphi $ in  $(2\lambda^2  \bs E^{n+1/2}, \bs \varphi )$ by the curl-free basis functions in Eq. \eqref{eq: C2}. It is worthwhile to note that the computational cost of employing the preconditioner is negligible, as it does not involve the task of updating particles, which consumes the majority portion of the computational time.
\end{remark}

\section{Strang-splitting energy-conserving scheme}
\label{sec: strang-splitting}

In this section, we introduce another implicit energy-conserving scheme, with the idea stemming from the Strang operator-splitting method, for comparison with the proposed energy-conserving asymptotic-preserving scheme. Similar techniques have been explored for Vlasov-Maxwell system \cite{2022Zhong,doi:10.1137/120871791}. Nevertheless, for the electrostatic Vlasov system, we emphasize the necessity of using curl-free spatial discretization to achieve energy conservation.

The particle-Amp\'ere system \eqref{eq: discrete-dimensionless-motion-Ampere} is firstly decomposed into the following two subproblems:
\begin{equation}\label{eq: subproblem1}
\begin{cases}
& \text{find}\; \{\bs x_p, \;\bs v_p \}_{p=1}^{N}\in \mathbb{R}^d\; \text{and}\; \bs E\in \bs H^{\rm curl,0},\; \text{ such that},\\
 &\dfrac{d\bs{x}_{p}}{dt}=\bs{v}_{p}, \\[6pt]
 &\dfrac{d\bs{v}_{p}}{dt}=\bs{0}, \\[6pt]
 &\lambda^{2}(\partial_{t}\bs{E},\bs{\varphi})=0, \;\; {\forall \bs \varphi \in \bs H^{\rm curl,0}},
\end{cases}
\end{equation}
and
\begin{equation}\label{eq: subproblem2}
\begin{cases}
& \text{find}\; \{\bs x_p, \;\bs v_p \}_{p=1}^{N}\in \mathbb{R}^d\; \text{and}\; \bs E\in \bs H^{\rm curl,0},\; \text{ such that},\\
    &\dfrac{d\bs{x}_{p}}{dt}=\bs{0},\\[6pt]
    &\dfrac{d\bs{v}_{p}}{dt}=-\int_{\Omega_{\bs{x}}} \bs{E}(\bs{x})S(\bs{x}-\bs{x}_{p})d\bs{x}, \\[6pt]
    &\lambda^{2}( \partial_{t}\bs{E}, \bs \varphi)+(\bs{J}, \bs \varphi)=\bs 0,\;\; {\forall \bs \varphi \in \bs H^{\rm curl,0}}.
\end{cases}
\end{equation}
It is direct to prove that each of the subproblems satisfies the energy conservation law
\begin{equation}
    \displaystyle \frac{d}{dt}\left[\frac{\lambda^{2}}{2}\int|\bs{E}(\bs{x})|^{2}d\bs{x}+\frac{w_{p}}{2}\sum\limits_{p=1}\limits^{N_{s}}\bs{v}_{p}^{2}\right]=0.
    \label{eq: strang-conservation}
\end{equation}
In light of the Strang operator-splitting framework, the above two subproblems can be discretized as follows:
\begin{equation}\label{eq: subproblemdis1}
\text{Step 1:\ } \begin{cases}
& \text{find}\; \{\bs x_p^{*}, \;\bs v_p^{*} \}_{p=1}^{N}\in \mathbb{R}^d\; \text{and}\; \bs E^{*}\in \mathbb{C}_N^d,\; \text{ such that},\\[3pt]
 &\dfrac{\bs{x}_{p}^{*}-\bs x_p^n}{\Delta t /2}= \dfrac{ \bs{v}_{p}^*+\bs{v}_{p}^n}{2} , \\[8pt]
 &\dfrac{\bs{v}_{p}^{*}-\bs v_p^n}{\Delta t /2}=\bs{0}, \\[6pt]
 &\dfrac{\lambda^{2}}{\Delta t /2}(\bs{E}^{*}-\bs E^n,\bs{\varphi})=0, \;\; {\forall \bs \varphi \in \mathbb{C}_N^d},
\end{cases}
\end{equation}
\begin{equation}\label{eq: subproblemdis2}
\text{Step 2:\ }\begin{cases}
& \text{find}\; \{\bs x_p^{**}, \;\bs v_p^{**} \}_{p=1}^{N}\in \mathbb{R}^d\; \text{and}\; \bs E^{**}\in \mathbb{C}_N^d,\; \text{ such that},\\
    &\dfrac{\bs{x}_{p}^{**}-\bs x_p^*}{\Delta t}=\bs{0},\\[6pt]
    &\dfrac{\bs{v}_{p}^{**}-\bs v_p^*}{\Delta t}=- \displaystyle\int \dfrac{\bs E^{**}+\bs E^*}{2}S\left(\bs{x}-\dfrac{\bs x_p^{**}+\bs x_p^{*}}{2}\right)d\bs{x}, \\[6pt]
    &\dfrac{\lambda^{2}}{\Delta t}(\bs E^{**}-\bs E^*, \bs \varphi)+\big(\bs J^{**}, \bs \varphi \big)=\bs 0,\;\; {\forall \bs \varphi \in \mathbb{C}_N^d},
\end{cases}
\end{equation}
with $\bs J^{**}=-\sum\limits_{p=1}\limits^{N_{s}}w_{p}S\left(\bs{x}-\dfrac{\bs x_p^{**}+\bs x_p^{*}}{2}\right)\dfrac{\bs v_p^{**}+\bs v_p^{*}}{2}$ and
\begin{equation}\label{eq: subproblemdis1prime}
\text{Step 3:\ }\begin{cases}
& \text{find}\; \{\bs x_p^{n+1}, \;\bs v_p^{n+1} \}_{p=1}^{N}\in \mathbb{R}^d\; \text{and}\; \bs E^{n+1}\in \mathbb{C}_N^d,\; \text{ such that},\\[3pt]
 &\dfrac{\bs{x}_{p}^{n+1}-\bs x_p^{**}}{\Delta t /2}= \dfrac{ \bs{v}_{p}^{n+1}+\bs{v}_{p}^{**}}{2} , \\[8pt]
 &\dfrac{\bs{v}_{p}^{n+1}-\bs v_p^{**}}{\Delta t /2}=\bs{0}, \\[6pt]
 &\dfrac{\lambda^{2}}{\Delta t /2}(\bs{E}^{n+1}-\bs E^{**},\bs{\varphi})=0, \;\; {\forall \bs \varphi \in \mathbb{C}_N^d},
\end{cases}
\end{equation}

Following the same procedure as in Theorem \ref{thm: thmec}, one directly obtains the discrete energy conservation law
\begin{equation}
\begin{aligned}
& \frac{1}{2}\lambda^{2}\int_{\Omega_{\bs{x}}}|\bs{E}^{n+1}|^{2}d\bs{x}+\sum\limits_{p=1}\limits^{N_{s}}\frac{1}{2}w_{p}|\bs{v}_{p}^{n+1}|^{2}= \frac{1}{2}\lambda^{2}\int_{\Omega_{\bs{x}}}|\bs{E}^{**}|^{2}d\bs{x}+\sum\limits_{p=1}\limits^{N_{s}}\frac{1}{2}w_{p}|\bs{v}_{p}^{**}|^{2} \\
&= \frac{1}{2}\lambda^{2}\int_{\Omega_{\bs{x}}}|\bs{E}^* |^{2}d\bs{x}+\sum\limits_{p=1}\limits^{N_{s}}\frac{1}{2}w_{p}|\bs{v}_{p}^*|^{2}   =    \frac{1}{2}\lambda^{2}\int_{\Omega_{\bs{x}}}|\bs{E}^n|^{2}d\bs{x}+\sum\limits_{p=1}\limits^{N_{s}}\frac{1}{2}w_{p}|\bs{v}_{p}^n|^{2}.
\end{aligned}
\label{eq: strang-energy-stable}
\end{equation}
The second and third therms of Eq. \eqref{eq: strang-energy-stable} have such implication: the total energy remains stable in each substep, given the fact that $\frac{1}{2}\lambda^{2}\int_{\Omega_{\bs{x}}}|\bs{E}^{**}|^{2}d\bs{x}+\sum\limits_{p=1}\limits^{N_{s}}\frac{1}{2}w_{p}|\bs{v}_{p}^{**}|^{2}$ is the energy of Step 2 and $\frac{1}{2}\lambda^{2}\int_{\Omega_{\bs{x}}}|\bs{E}^* |^{2}d\bs{x}+\sum\limits_{p=1}\limits^{N_{s}}\frac{1}{2}w_{p}|\bs{v}_{p}^*|^{2} $ the energy of Step 3.

Compared with the fully-implicit scheme proposed in the previous section, the computation of particles and the electric field in the Strang splitting scheme can be decoupled and the first two steps can be combined to arrive at the following efficient solution algorithm:
\begin{itemize}
\item[] Step 1: find $ \bs {\tilde E}=(\bs E^{**}+\bs E^n)/2  \in \mathbb{C}_N^d$ such that
\begin{equation}\label{eq: spE}
  \displaystyle {2\lambda^{2}}(\bs {\tilde E}-\bs E^*, \bs \varphi)+{\Delta t}\big(\bs {\tilde J}  , \bs \varphi \big)=\bs 0,\;\; {\forall \bs \varphi \in \mathbb{C}_N^d},
\end{equation}
with
\begin{equation}
\bs {\tilde J} = -\sum\limits_{p=1}\limits^{N_{s}}w_{p}S\left(\bs{x}-\bs x_p^n-\frac{\Delta t}{2} \bs v_p^n \right) \left[ \bs v_p^n-\frac{\Delta t}{2} \int \bs {\tilde E}(\bs{x})S\left(\bs x-\bs x_p^n-\frac{\Delta t}{2}\bs v_p^n\right)d\bs x \right]
\end{equation}
and compute
\begin{equation}
\bs v_p^{**}= \bs v_p^n-\frac{\Delta t}{2} \int \bs {\tilde E}(\bs{x})S\left(\bs x-\bs x_p^n-\frac{\Delta t}{2}\bs v_p^n\right)d\bs x,\quad \bs x_p^{**}=\bs x_p^n+\frac{\Delta t}{2}\bs v_p^n.
\end{equation}

\item[] Step 2: Compute
\begin{equation}
\bs v_p^{n+1}=\bs v_p^{**},\quad \bs x_p^{n+1} =\bs x_p^{**}+\frac{\Delta t}{2}\bs v_p^{**},\quad \bs E^{n+1}=2\bs {\tilde  E}-\bs E^n.
\end{equation}
\end{itemize}
Note that the nonlinear system in Eq. \eqref{eq: spE} can be solved efficiently by the Anderson acceleration algorithm preconditioned with the diagonal matrix $\mathcal{D}$ defined in Remark \ref{rem: Dprecond}.

\section{Numerical results}
\label{sec: numerical}

We perform numerical results to validate the performance of the fully-implicit energy-conserving  scheme with the asymptotic-preserving preconditioner (dubbed as ``AP-EC" scheme)
proposed in Section \ref{sec: fully-implicit}, and the Strang operator-splitting method (dubbed as ``SS-EC" scheme) presented in Section \ref{sec: strang-splitting}.
We also show the results of the classical leapfrog scheme for comparisons. All simulations are carried out in phase space $\Omega_{\bs{x}}\times\Omega_{\bs{v}}$ with two dimensions in space and two dimensions in velocity. The computation domain is $\Omega_{\bs x}=[0,2\pi/k_{1}]\times[0,2\pi/k_{2}]$, where $\bs{k}=(k_{1},k_{2})^T$ are given constants. Moreover, let $L_{x}$ and $L_{y}$ be the length of the space domain. We consider the one-species system. The computational domain is uniformly discretized into $32\times32$ cells.
The benchmark problems include Landau damping, two-stream instability, and bump-on-tail instability.
In each problem, $10^5$ macro particles, which are subject to the corresponding initial distributions, are used.

\subsection{Landau damping}
Landau damping has been widely investigated in plasma physics \cite{1965445,landau2011,Stix1962}. In the 2D Landau damping problem, a small cosine perturbation with an amplitude $\alpha$ is exerted on uniformly distributed particles,
and the initial velocity of the particles obeys a Maxwellian distribution (we take the thermal velocity to be 1):
\begin{equation}
    \displaystyle f_{0}(\bs{x},\bs{v})=\frac{1}{2\pi L_{x}L_{y}}\left(1+\alpha\cos(\bs{k}\cdot\bs{x})\right)e^{-\frac{\bs{v}^{2}}{2}}.
\end{equation}
We take $\alpha=0.1$ and $\bs{k}=(0.3,0.3)^{T}$ in our calculations.

\begin{figure}[!htb]
    \centering
    \includegraphics[width=0.45\textwidth]{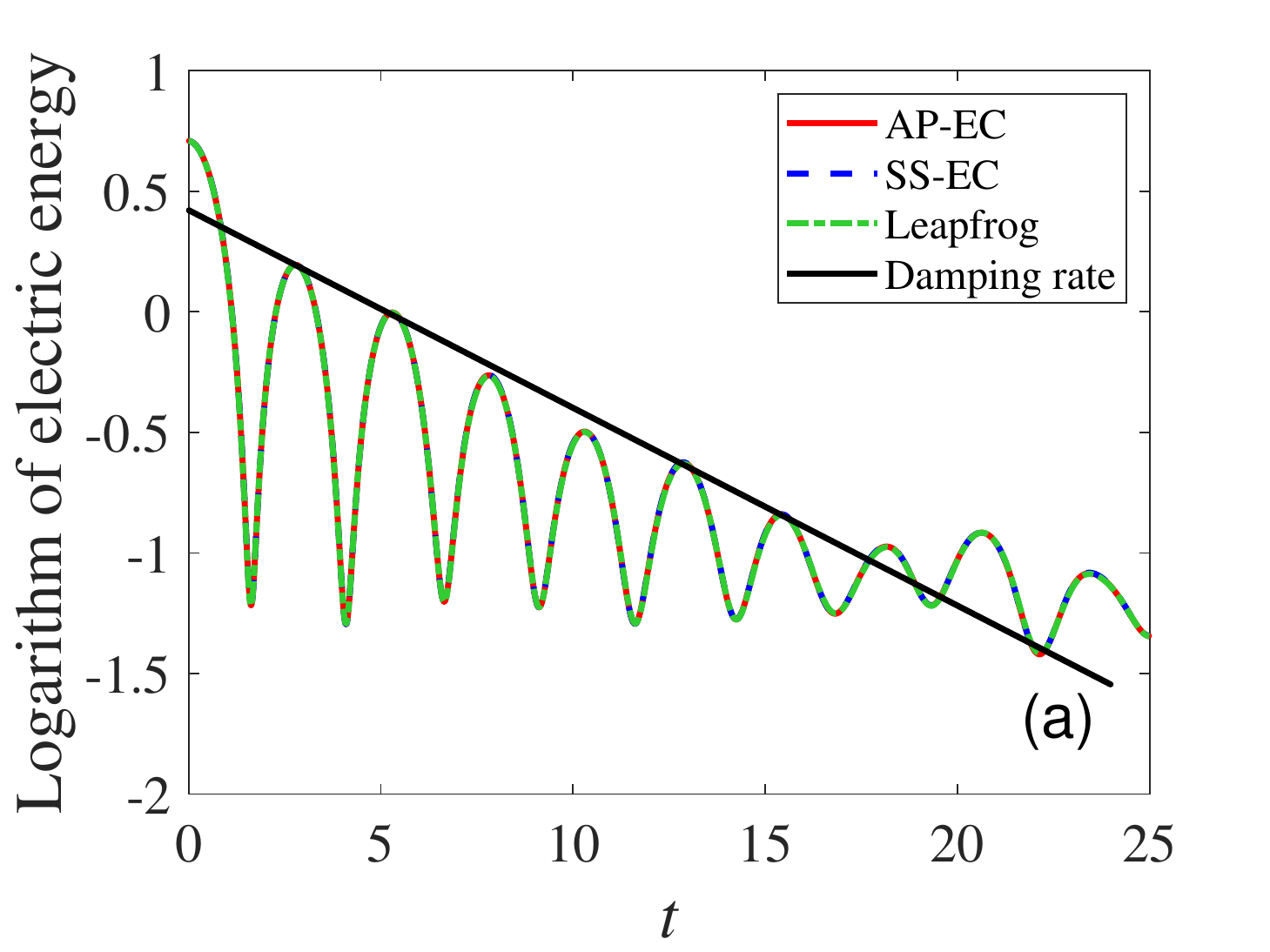}
    \includegraphics[width=0.45\textwidth]{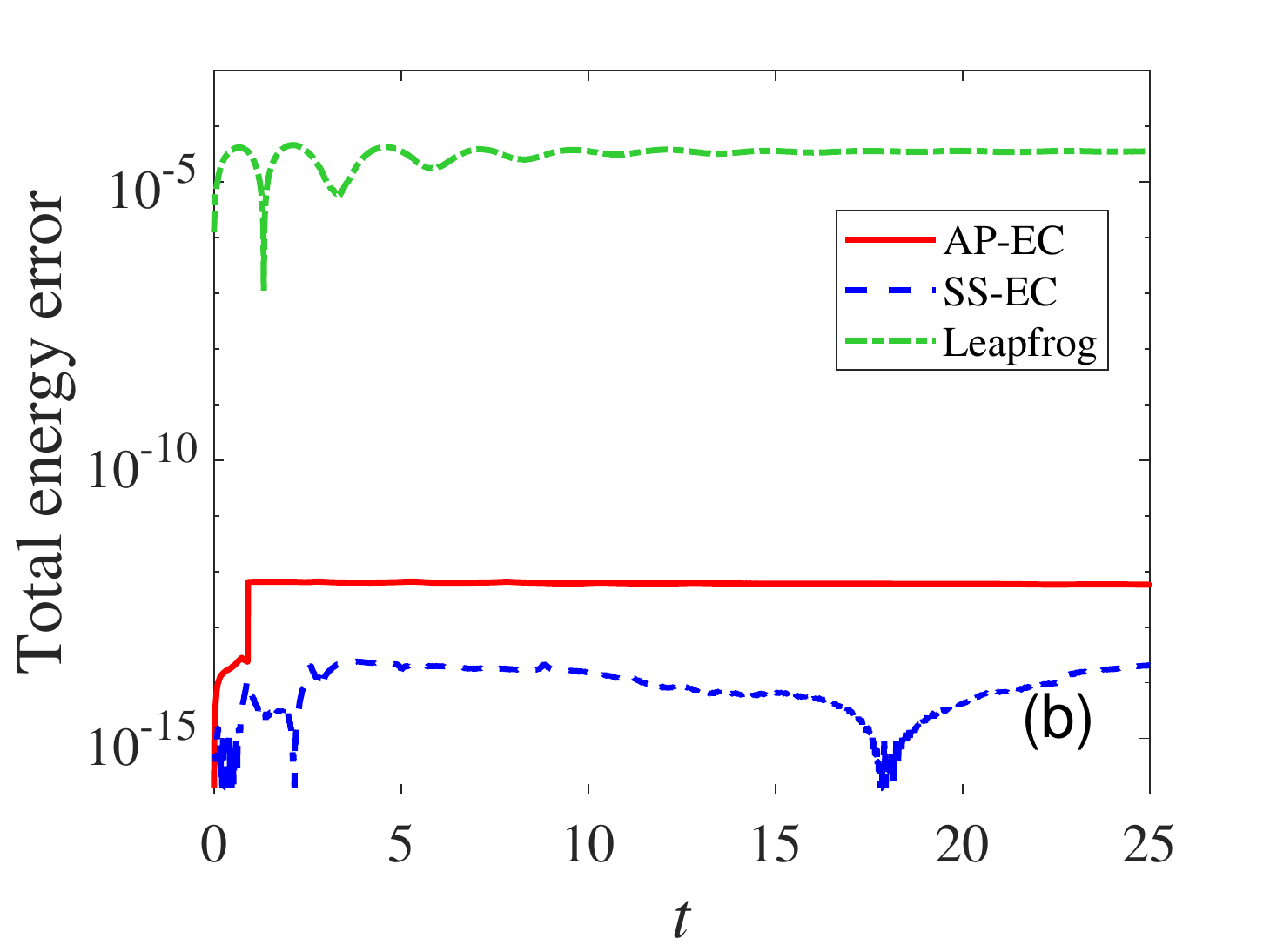}
    \caption{2D Landau damping with $\lambda=1$ and $\Delta t=0.01$: (a) Electric energy in logarithmic scale; and (b) Relative error of total energy.}
    \label{fig: landau-conservation-test}
\end{figure}

We first calculate the results with $\lambda=1$ and time step $\Delta t=0.01$. The tolerances of AA iterations in nonlinear field equations
and the particle pusher are set to be $10^{-11}$ and $10^{-10}$, respectively.
Fig. \ref{fig: landau-conservation-test} presents the electric energy and the relative error of the total energy for the three methods:
the AP-EC, the SS-EC and the classical leapfrog schemes. The total energy of the system is 445.5521. One can observe that all these three methods
remain stable even after a long simulation, and predict the Landau damping well for $\sim 10$ periods. The AP-EC and SS-EC schemes have
a relative error of the total energy less than $10^{-12}$, whereas that of the classical leapfrog scheme
is at the level of $10^{-5}$. These results demonstrate that both the AP-EC and SS-EC are energy-conserving schemes.

In the following calculations, we set the AA-iteration tolerances for the field equations and the particle pusher to be $10^{-6}$ and $10^{-9}$, respectively.
One can observe that the relaxation of the error tolerances will not affect energy conservation considerably, and the relative error of the total energy
in the implicit schemes remains at the level of $10^{-8}\sim10^{-9}$. In Fig. \ref{fig: l1-dt001-k03-curlE-gausslaw},
one displays the residuals of the curl-free constraint $\nabla\times\bs{E}=\bs{0}$ and
the Gauss law $\lambda^{2}\nabla\cdot\bs{E}=1-n$, where the maximum value of the residual absolute is measured.
These results demonstrate that the discrete curl-free condition is strictly satisfied for the AP-EC and the SS-EC schemes.
However, the Gauss law is not exactly preserved in two implicit schemes. The residual of the AP-EC is at the level
of $10^{-6}$, much better than that of the SS-EC, which is of $\mathcal{O}(10^{-4})$.

\begin{figure}[!htb]
    \centering
    \includegraphics[width=0.45\textwidth]{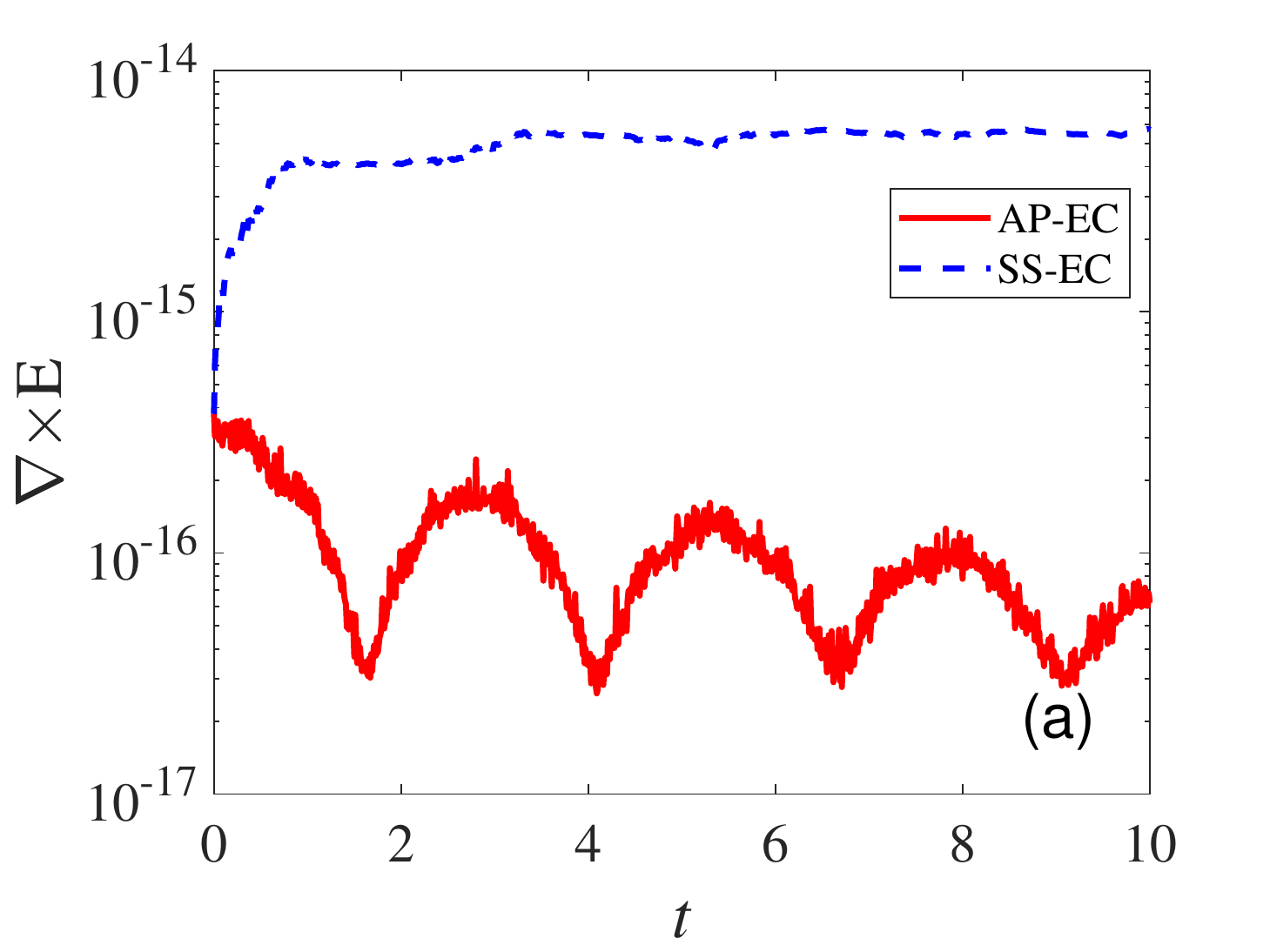}
    \includegraphics[width=0.45\textwidth]{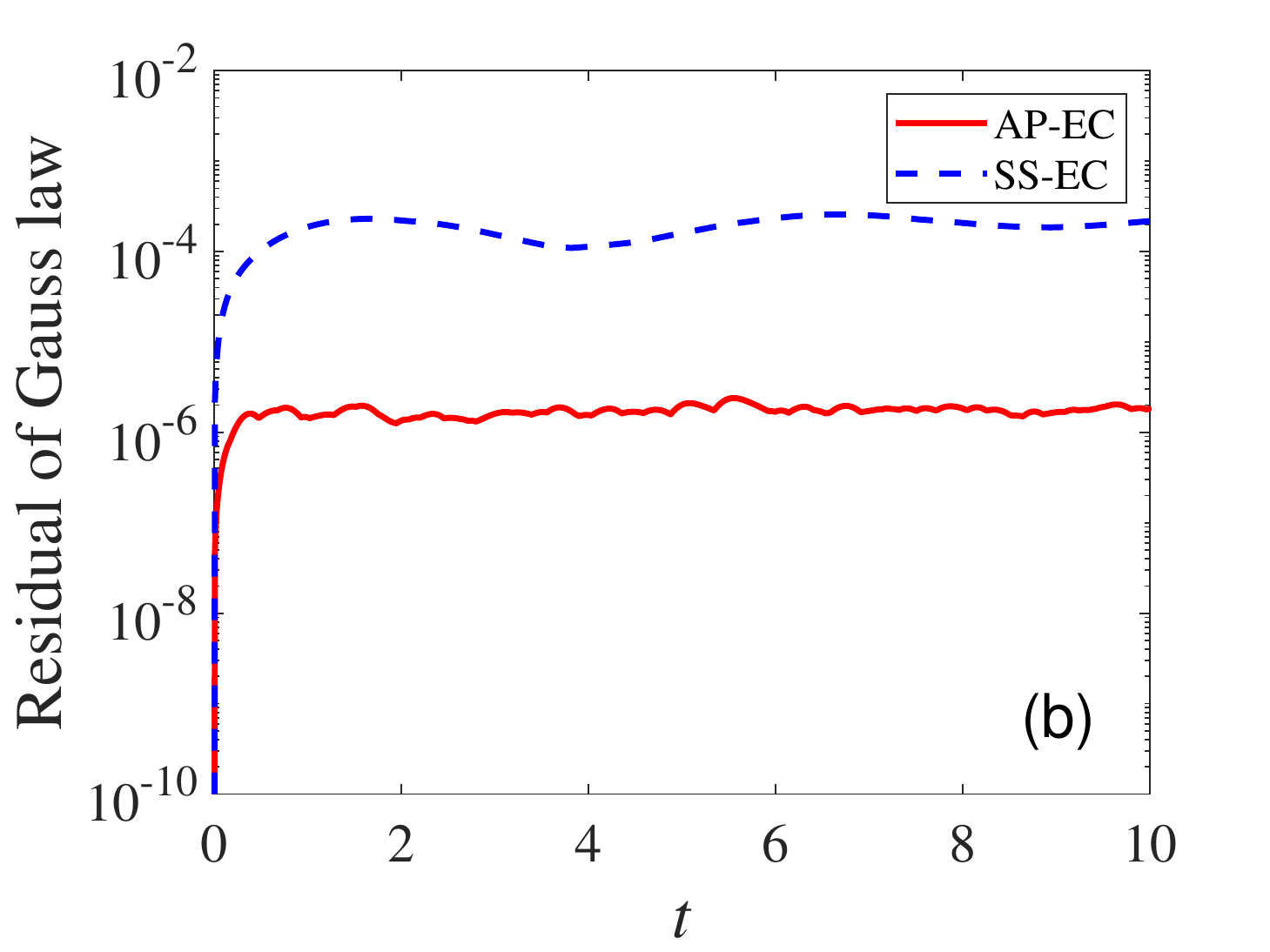}
    \caption{The preservation of curl-free electric field and residual of Gauss law for Landau damping with $\lambda=1$ and $\Delta t=0.01$.
    (a) Maximum value of $|\nabla\times\bs{E}|$ on grid points; (b) Residual of Gauss law. }
    \label{fig: l1-dt001-k03-curlE-gausslaw}
\end{figure}

\begin{figure}[!htb]
    \centering
    \includegraphics[width=0.45\textwidth]{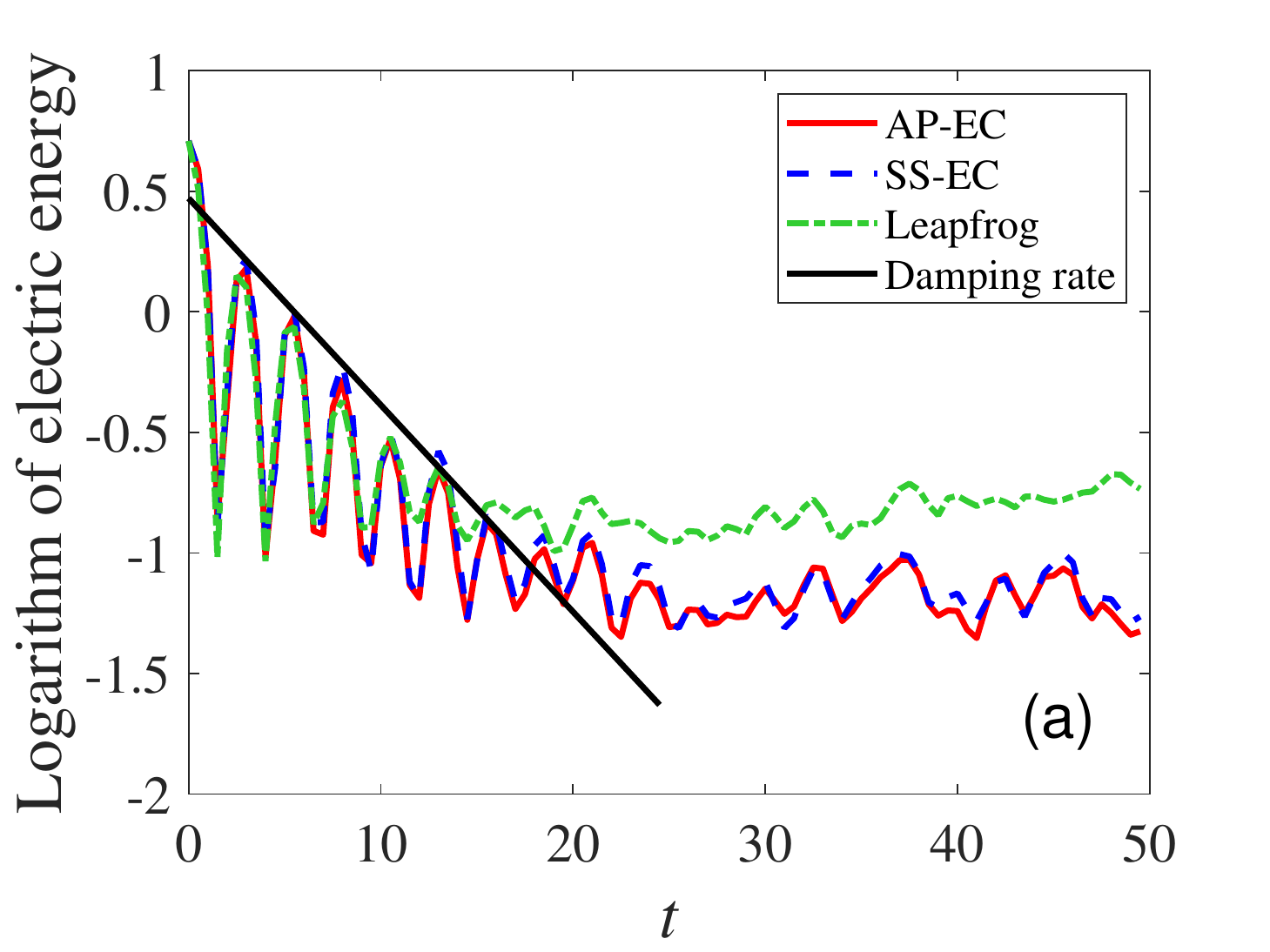}\label{fig: landau-l1-dt05-k03-WE}
    \includegraphics[width=0.45\textwidth]{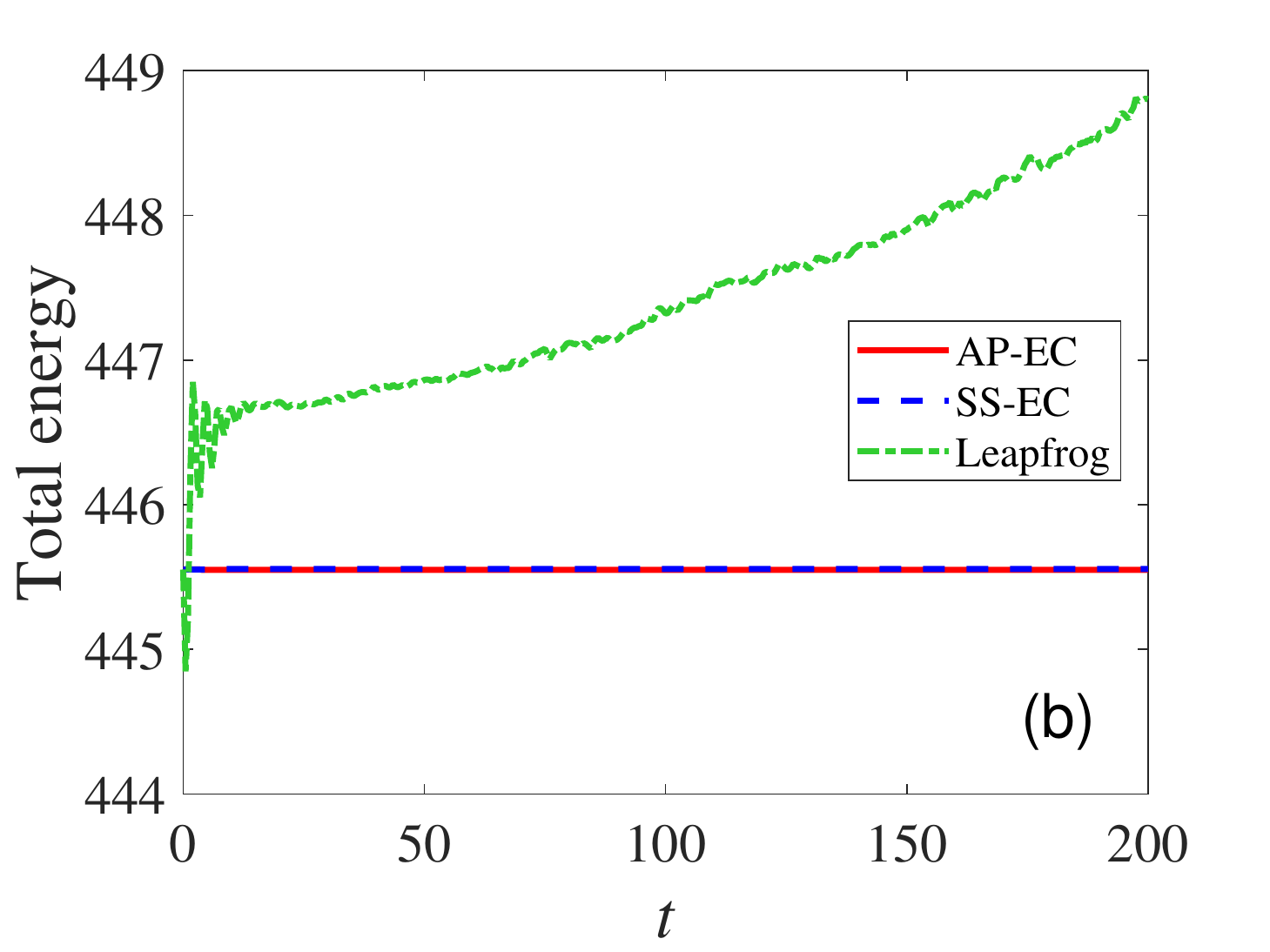}\label{fig: landau-l1-dt05-k03-TE}
    \caption{2D Landau damping with $\lambda=1$ and a large time step $\Delta t=0.5$.
    (a) Electric energy in logarithmic scale; (b) Total energy.}
    \label{fig: landau-l1-dt05-k03-WE-TE}
\end{figure}

We now conduct a long-duration simulation with large time step as $\Delta t=0.5$ with the same $\lambda=1$,
and the results are presented in Fig. \ref{fig: landau-l1-dt05-k03-WE-TE}.
Panel (a) shows that the classical leapfrog scheme cannot withstand large-time-step simulation, which deviates the damping rate after a few periods, while the AP-EC and SS-EC maintain
excellent performance. Additional calculations of the leapfrog scheme with $\Delta t=0.01$ (not shown in the figure) agree well with the curves of the AP-EC and SS-EC, demonstrating the stability of the
implicit schemes. Furthermore, panel (b) illustrates that the AP-EC and SS-EC schemes are energy-conserving for large time steps, but the deviation of the leapfrog scheme from the exact energy increases with the time.

\begin{figure}[!htb]
    \centering
     \includegraphics[width=0.45\textwidth]{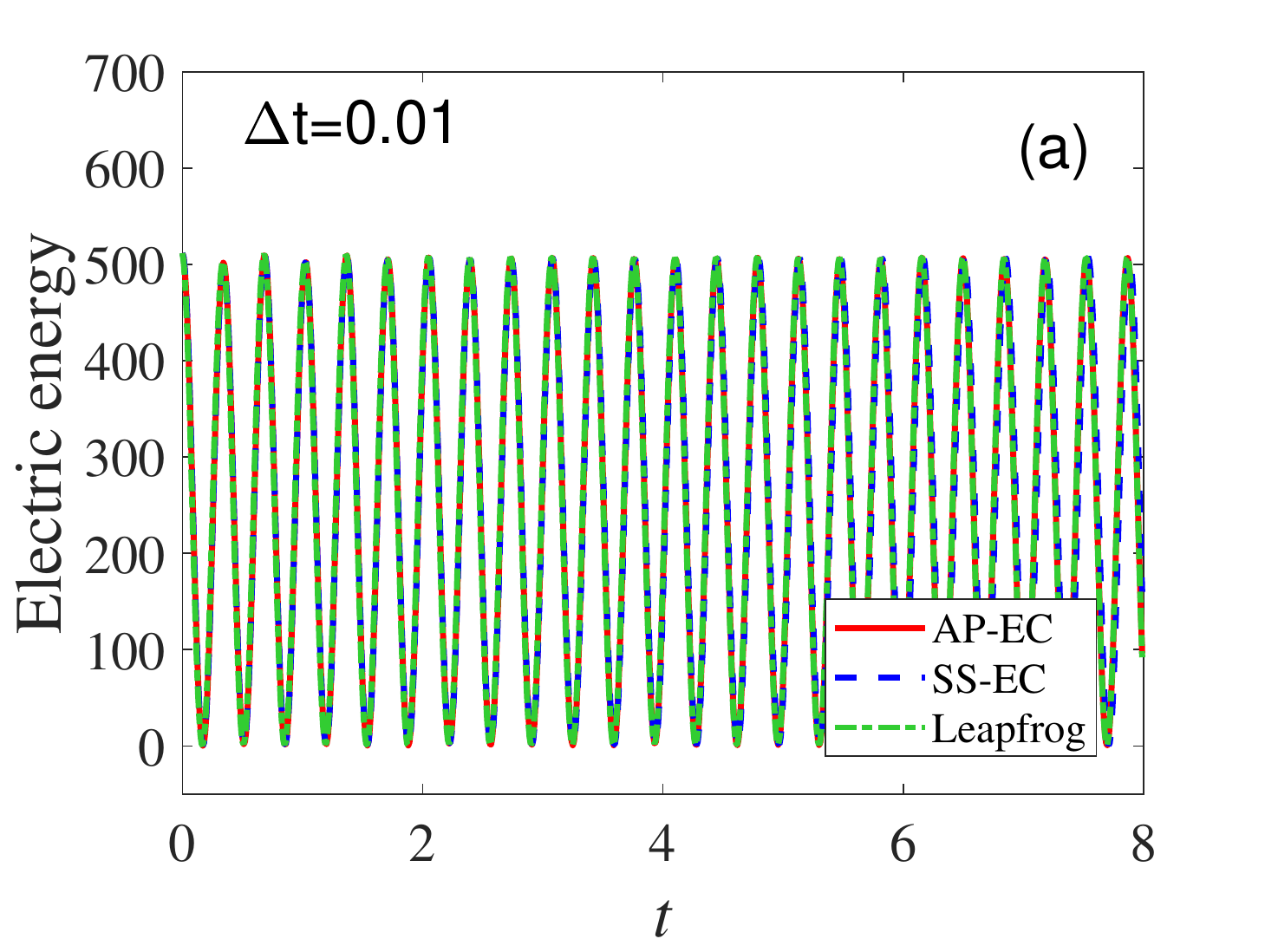}
    \includegraphics[width=0.45\textwidth]{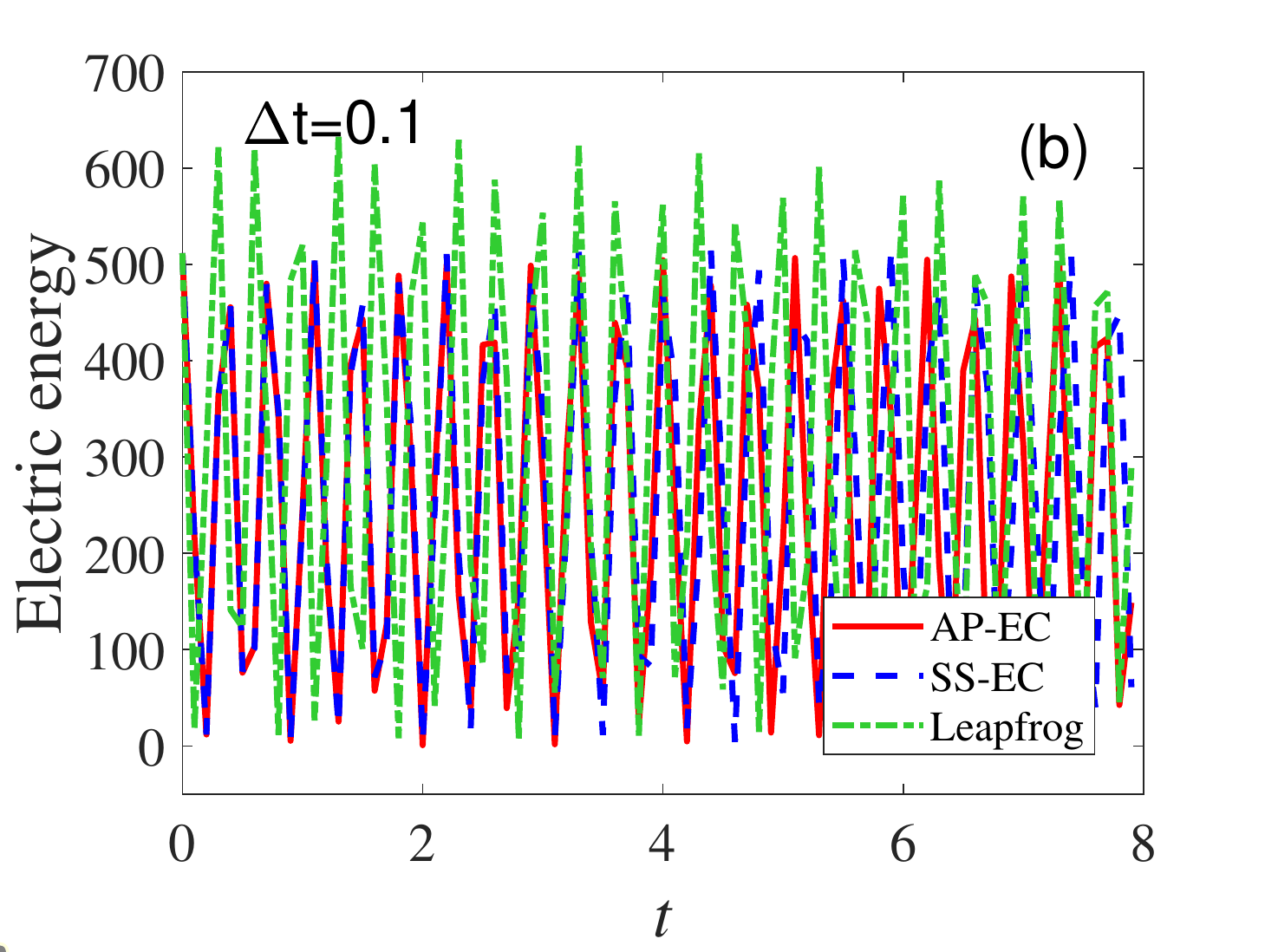}

    \includegraphics[width=0.45\textwidth]{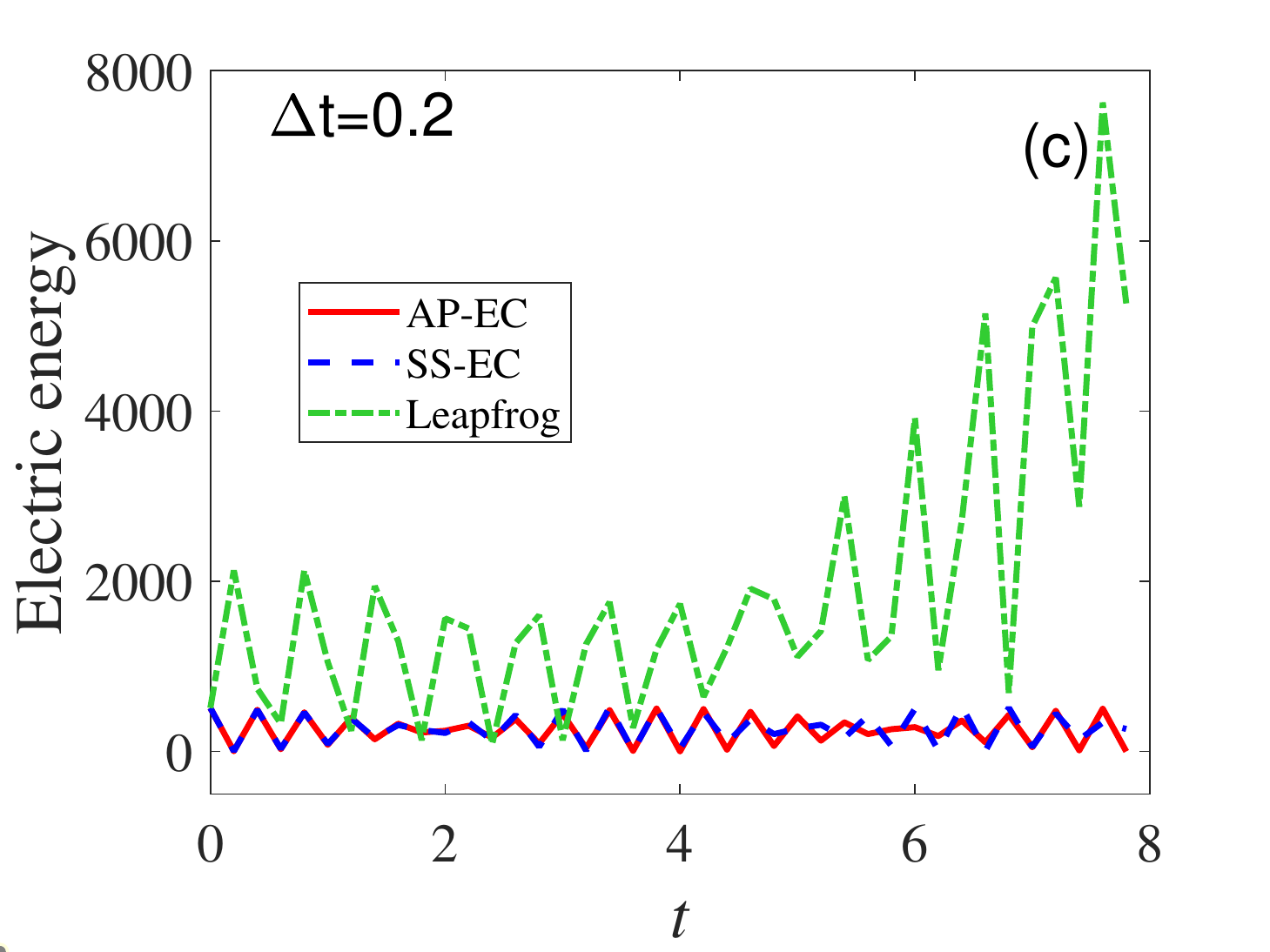}
    \includegraphics[width=0.45\textwidth]{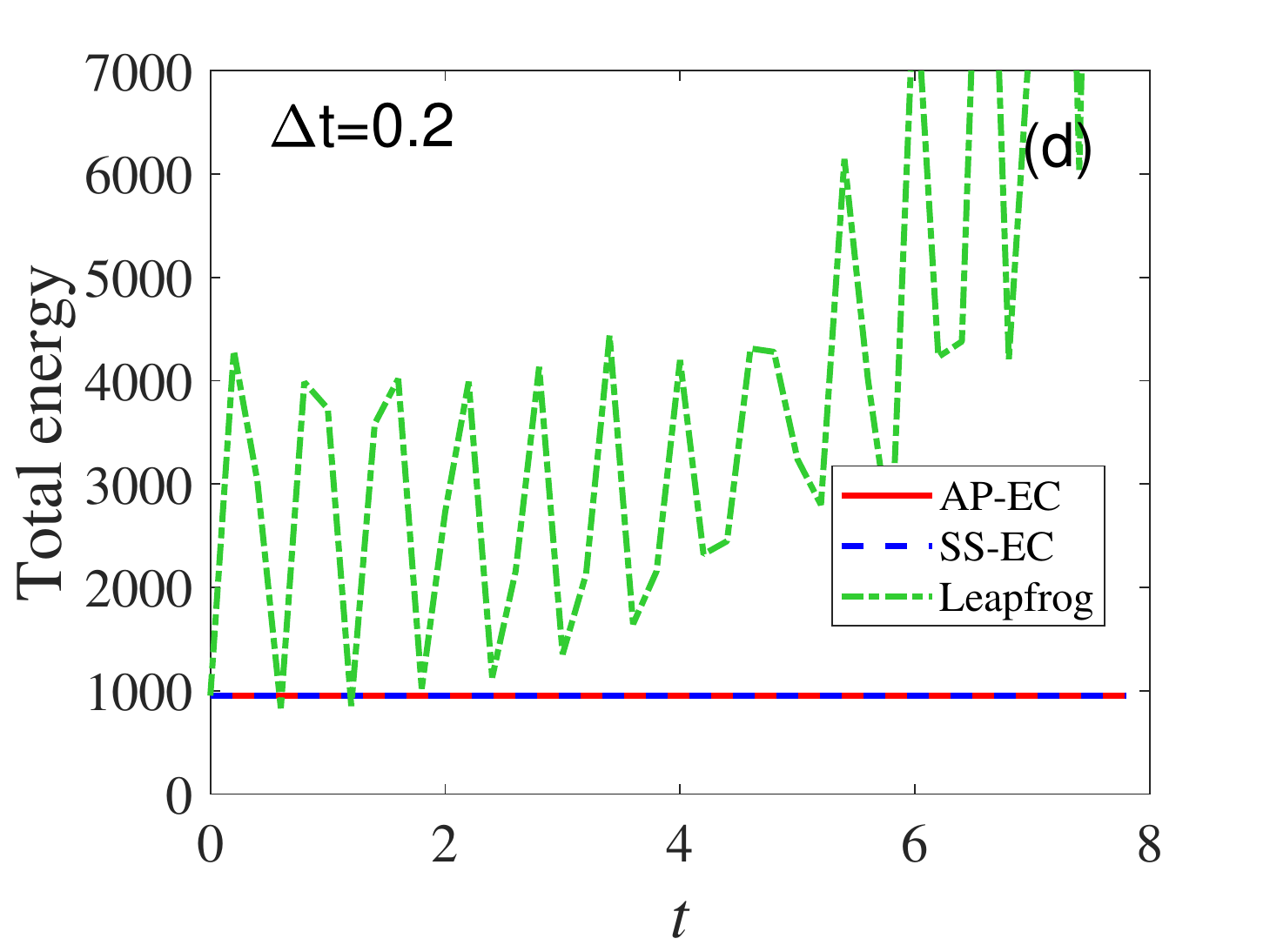}
    \caption{2D Landau damping with dimensionless Debye length $\lambda=0.1$. (a)~Electric energy with $\Delta t=0.01$; (b)~Electric energy with $\Delta t=0.1$; (c)~Electric energy with $\Delta t=0.2$; and (d)~Total energy with $\Delta t=0.2$.}
    \label{fig: landau-l01-dt01-dt02}
\end{figure}

For the performance of our schemes in different physical scales, simulations for $\lambda<1$ are also conducted. Fig. \ref{fig: landau-l01-dt01-dt02} displays the results of cases close to the quasi-neutral limit with dimensionless Debye length $\lambda=0.1$ and with three time steps $\Delta t=0.01, 0.1$ and  $0.2$. For the case of $\Delta t=0.01$ shown in panel (a), curves calculated by the three methods are almost overlapping, and there are three damping periods in each unit of time.  For larger time steps, however, the classical leapfrog scheme has already presented substantial phase differences with $\Delta t=0.1$ and predicts larger electric energy, and it becomes blowing up with $\Delta t=0.2$. We observe that implicit methods are energy-conserving for $\Delta t=0.2$,  and the relative error of total energy for the AP-EC and SS-EC schemes are both at the order of $\mathcal{O}(10^{-8})$.
Additionally, the peaks of the fast Langmuir oscillation are nearly equal from small to large time steps. Regarding troughs, the AP-EC results of $\Delta t=0.1$ are the closest to the results of $\Delta t=0.01$. The SS-EC results are slightly larger than those of $\Delta t=0.01$. The cyclical variation of magnitudes of the AP-EC and SS-EC in Fig. \ref{fig: landau-l01-dt01-dt02} with $\Delta t=0.2$ is attributed to the Shannon sampling theorem and the aliasing error \cite{MDFTWEB07,aliasing,e14112192,Shannon}.

The test case of the 2D Landau damping is ended with Table \ref{tab: average-number-AA-iteration} on the average AA iteration times before convergence to show the performance of the proposed asymptotic-preserving preconditioner. For comparison, we replace the asymptotic-preserving preconditioner in the AP-EC scheme with the preconditioner $\mathcal{D}$ defined in Remark \ref{rem: Dprecond} and term the resultant method as ``$\mathcal{D}$-EC'' scheme. It can be discerned from Table \ref{tab: average-number-AA-iteration} that for the case of $\lambda=1$, the average iteration time of the AP-EC is small, and it only increases slightly from ~5 to 7 when the time step sizes are magnified by 20 times. The $\mathcal{D}$-EC takes more iterations to converge when $\lambda=1$, and the difference of iteration times between the AP-EC and the $\mathcal{D}$-EC grows as $\Delta t$ becomes larger.
While for the case of $\lambda=0.1$, the system approaches the quasi-neutral limit and the resultant nonlinear system becomes extremely difficult to solve (the density is magnified by 100 times compared with the case of $\lambda=1$ and the nonlinear effect dominates). It can be seen that the iteration time of the AP-EC is still tiny for $\Delta t=0.01$ and $\Delta t=0.1$. In contrast, for $\Delta t=0.2$, it takes more iterations before convergence, possibly due to the inaccuracy caused by the particle pusher under large time steps. On the contrary, the $\mathcal{D}$-EC converges for the case of $\Delta t=0.01$ slowly and fails to converge for the rest of the cases.
\begin{table}[H]
    \centering
    \begin{tabular}{c|c|c}
    \toprule[2pt]
    Parameters & AP-EC   & $\mathcal{D}$-EC  \\\midrule[1.5pt]
     $\lambda=1$, $\Delta t=0.01$, tol=1e-6    & 5.130 & 7.010 \\\midrule[1pt]
     $\lambda=1$, $\Delta t=0.1$, tol=1e-6    & 6.560 & 10.67\\\midrule[1pt]
     $\lambda=1$, $\Delta t=0.2$, tol=1e-6    &  7.420 &  15.65 \\\midrule[1pt]
     $\lambda=0.1$, $\Delta t=0.01$, tol=1e-6    &  7.620 &  21.7222
\\\midrule[1pt]
     $\lambda=0.1$, $\Delta t=0.1$, tol=1e-6    &20.32
  & not conv.\\\midrule[1pt]
     $\lambda=0.1$, $\Delta t=0.2$, tol=1e-6  &  157.2 & not conv.\\
     \bottomrule[2pt]
    \end{tabular}
    \caption{Average AA iteration times~(per step) before convergence. The iteration depth is $m=13$, and the maximum number of AA iterations is 200 for each time step.}
    \label{tab: average-number-AA-iteration}
\end{table}

\subsection{Two-stream instability}
Next, we demonstrate the performance of our proposed schemes against the two-stream instability \cite{ghorbanalilu_abdollahzadeh_rahbari_2014,doi:10.1119/1.1407252}. The initial distribution is given by
\begin{equation}
    \displaystyle f_{0}(\bs{x},\bs{v},t)=\frac{1}{2\pi L_{x}L_{y}}\left(1+\alpha\cos(\bs{k}\cdot\bs{x})\right)\left[\frac{1}{2}e^{-\frac{(\bs{v}-\bs{v}_{0})^{2}}{2}}+\frac{1}{2}e^{-\frac{(\bs{v}+\bs{v}_{0})^{2}}{2}}\right],
\end{equation}
where the parameters are set to $\bs{k}=(0.3,0)^T$ and $\bs{v}_{0}=(3,0)^T$. In the forthcoming numerical results, we focus on the energy portion of $\bs{E}_{x}$, defined as $\int (\lambda^2 |\bs{E}_{x}|^2/2)dx$.

\begin{figure}[!htb]
    \centering
    \includegraphics[width=0.45\textwidth]{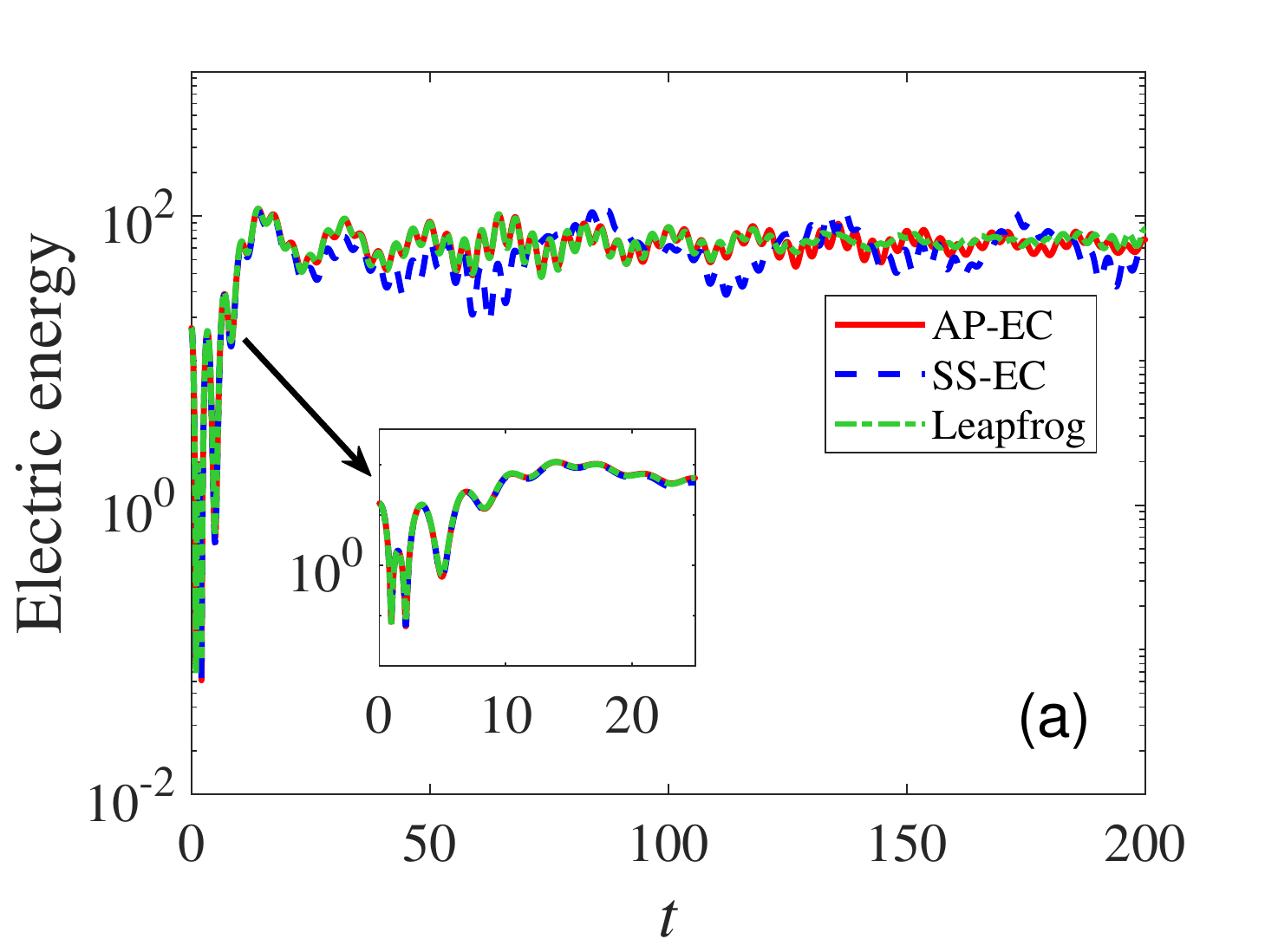}
    \includegraphics[width=0.45\textwidth]{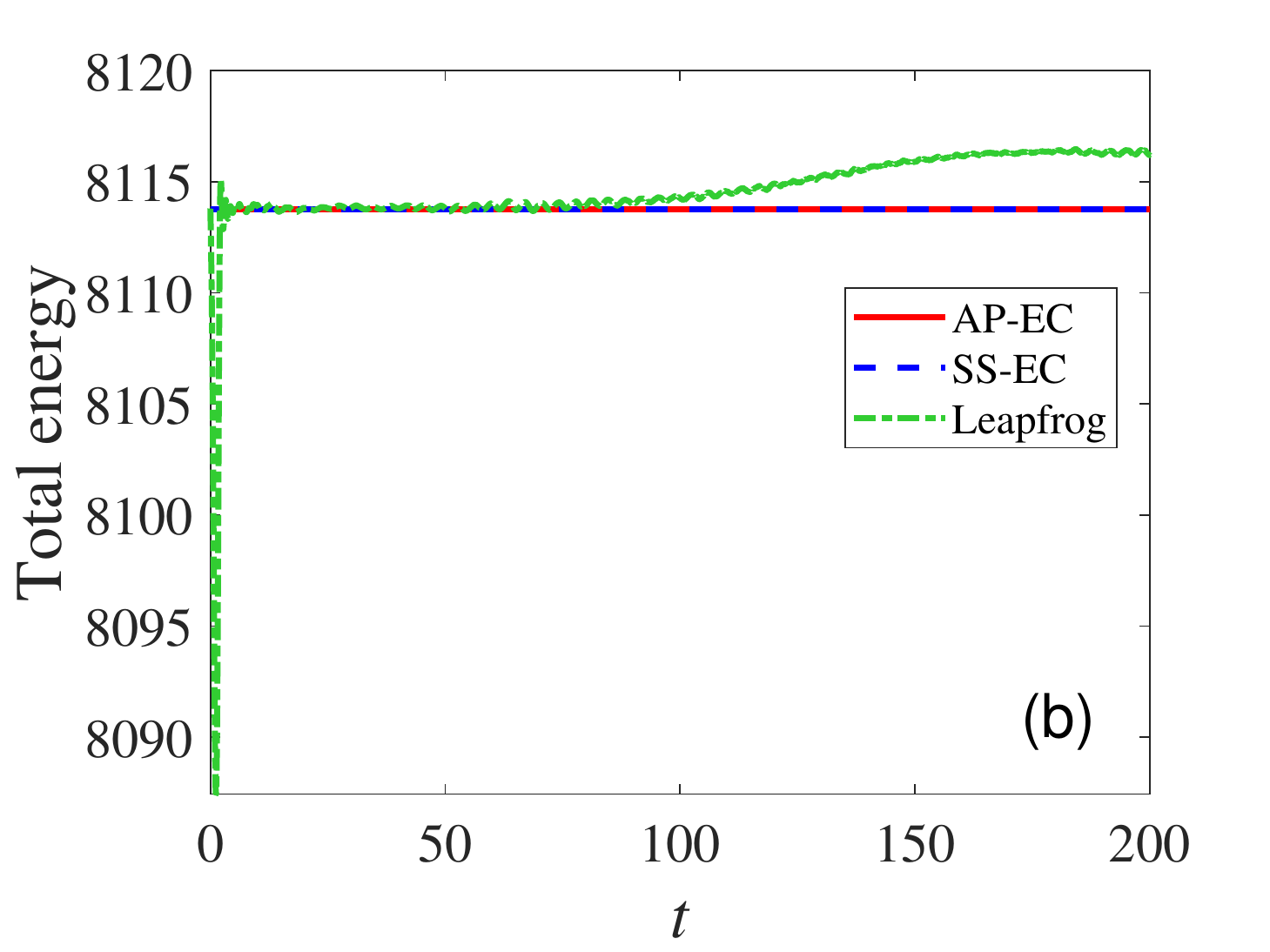}
    \caption{Two-stream instability with $\lambda=1$ and $\Delta t=0.01$: (a)~Electric energy of the $x$-component, (b)~Total energy.}
    \label{fig: twostream-l1-dt001-k03-20000-WE1-TE}
\end{figure}
Fig. \ref{fig: twostream-l1-dt001-k03-20000-WE1-TE} displays the results of the electric energy of the $x$-component and the total energy with $\lambda=1$ and $\Delta t=0.01$. Panel (a) shows that the SS-EC has a prominent error accumulation for long-time simulations, which is in consistent with the performance of Strang splitting scheme reported in other literature \cite{BOTCHEV2009522,Einkemmer2014ConvergenceAO}. Panel (b) shows the unconditional energy conservation of the AP-EC and SS-EC schemes. In contrast, there is a sharp deviation of total energy for the leapfrog method at the very beginning and an increased error in energy conservation with time.

The accumulated error of the SS-EC is also revealed in Fig. \ref{fig: twostream-l1-dt001-k03-10000-x-vx}, which is the phase space distribution of particles when $\lambda=1$ and $\Delta t=0.01$. We can conclude from this figure that with $\Delta t=0.01$, the ``eye" of $x-v_{x}$ by the AP-EC stays centered through the entire simulation, while a tiny movement of the ``eye" in the leapfrog is observed as the simulation proceeds. Significant bias occurs in the SS-EC at some point between $t=60$ and $t=80$. This phenomenon can be explained by the evolution of $\bs{E}_{x}$ in panel (a) of Fig. \ref{fig: twostream-l1-dt001-k03-20000-WE1-TE}. The electric field $\bs{E}_{x}$ calculated by the SS-EC starts to deviate from those of the AP-EC and leapfrog at around $t=50$, leading to differences in electric force. Hence errors exist in the $x-v_{x}$ phase space distribution.

\begin{figure}[!htb]
    \centering
    \includegraphics[width=1.05\textwidth]{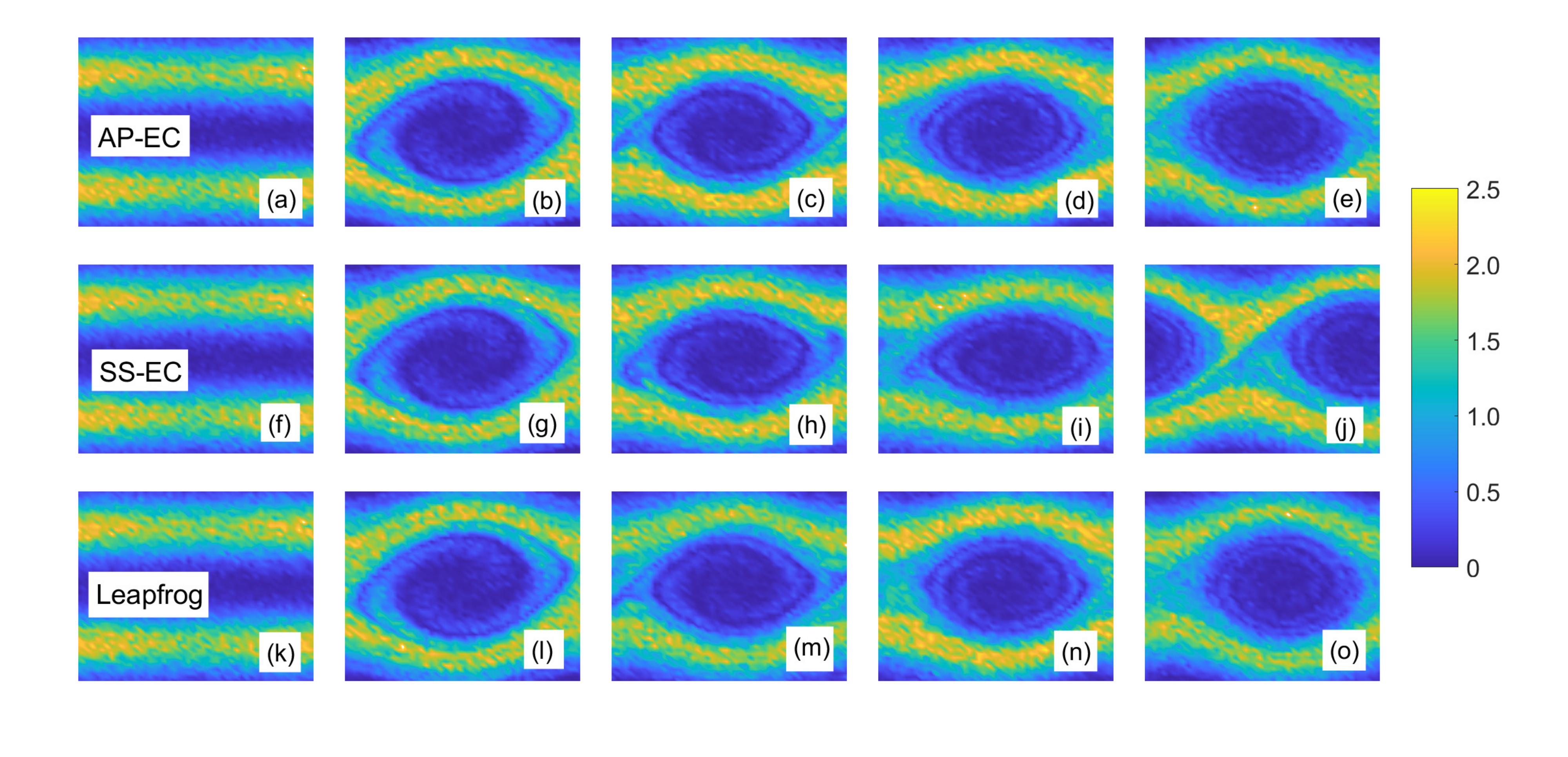}
    \caption{Two-stream instability with $\lambda=1$ and $\Delta t=0.01$. The $x-v_{x}$ phase space distribution at $t=0$, $t=20$, $t=40$, $t=60$ and $t=80$~(each column for one moment), (a-e): AP-EC; (f-j): SS-EC; (k-o): leapfrog. For all these figures, the $x$-axis is the $x$ positions of particles, the $y$-axis is the $x$-part of velocities, while $z$-axis~(color bar) denotes the distribution density of $(x,v_{x})$ in each cell of the $x-v_{x}$ plane. Range  of color bar is $[0,2.5]$.}
    \label{fig: twostream-l1-dt001-k03-10000-x-vx}
\end{figure}

We intend to conduct a series of simulations with smaller time steps to recalculate the phase space distribution at $t=80$ for the SS-EC scheme.
Fig. \ref{fig: twostream-convergence-strang-splitting} indicates that the results of the SS-EC will gradually coincide with that of the AP-EC~(shown in panel (d)). For the SS-EC, it requires a time step ten times smaller than that of the AP-EC to capture the evolution of macro particles accurately.

\begin{figure}[!htb]
    \centering
    \includegraphics[width=0.9\textwidth]{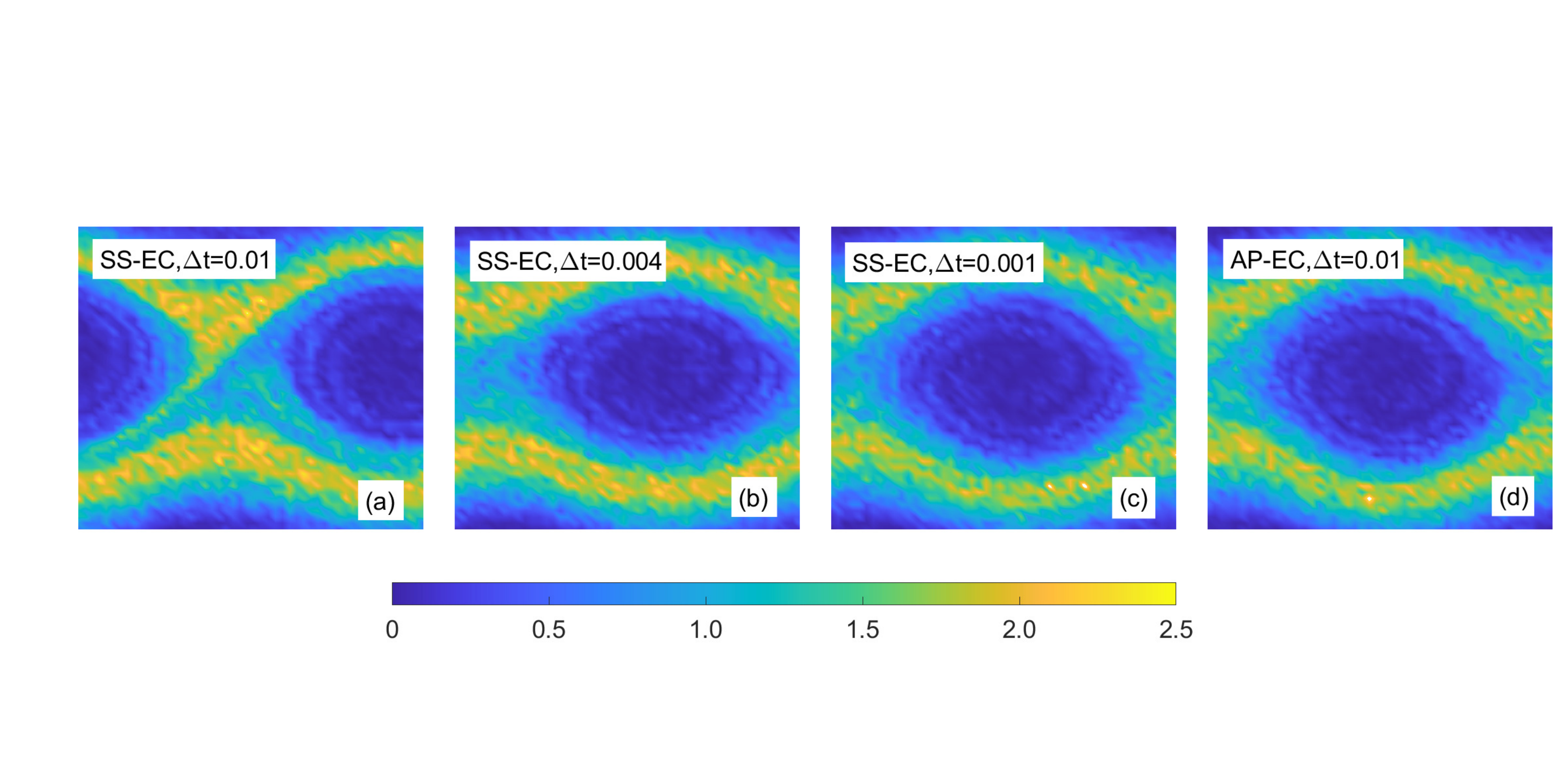}
    \caption{Convergence of $x-v_{x}$ phase space distribution for the Strang splitting scheme in two-stream instability at $t=80$. (a)~SS-EC with $\Delta t=0.01$; (b)~SS-EC with $\Delta t=0.004$; (c)~SS-EC with $\Delta t=0.001$; (d)~AP-EC with $\Delta t=0.01$.}
    \label{fig: twostream-convergence-strang-splitting}
\end{figure}

\subsection{Bump-on-tail instability}
Bump-on-tail instability is one of the fundamental and essential instabilities in plasma simulations \cite{Shoucri11}. Early numerical experiments have been conducted in the Vlasov-Poisson system \cite{Georg_Knorr_1977,Cheng1976TheIO,Gagn1977ASS}. In this test case, the system is initially perturbed both in $\bs{x}$ space and in $\bs{v}$ space, with the initial distribution function as \cite{kraus2017gempic}:
\begin{equation}
    \displaystyle f_{0}(\bs{x},\bs{v})=\frac{1}{2\pi L_{x}L_{y}}\left(1+\alpha\cos(\bs{k}\cdot\bs{x})\right)\left(\delta e^{-\frac{v_{1}^{2}}{2}}+2(1-\delta)e^{-2(v_{1}-v_{d})^{2}}\right)e^{-\frac{v_{2}^{2}}{2}}
    \label{eq: bump-on-tail-initial-distribution}
\end{equation}
in which we take $\bs{k}=(0.3,0)$, $\alpha=0.1$, $\delta=0.9$ and $v_{d}=3.5$.

We first focus on the time evolution of  the total energy, the $x$-component electric energy, and the  total electric energy~(Fig. \ref{fig: bump-on-tail-l1-dt001-k03-10000-g0}). From the time evolution of the total electric energy in panel (c) of Fig. \ref{fig: bump-on-tail-l1-dt001-k03-10000-g0}, we notice that the electric energy shows a reduction with time until about $t=2$, which can be explained numerically from the temporal decrease of the electric field $\bs{E}_{y}$. One can observe that the AP-EC and SS-EC agrees with each other in the electric energy of the $x$-component, and the leapfrog curve has an obvious deviation since $t=30$.
Fig. \ref{fig: bump-on-tail-l1-dt001-k03-10000-g0-distribution-vx} plots the distribution function of $v_{x}$. Curves of all the three schemes are in good accordance with our initialization in Eq. \ref{eq: bump-on-tail-initial-distribution}, for the reason that the maximum appears at $v_{d}=3.5$ with $v_{d}$ being the velocity of electron beam drift of our system.

\begin{figure}[!htb]
    \centering
    \includegraphics[width=0.3\textwidth]{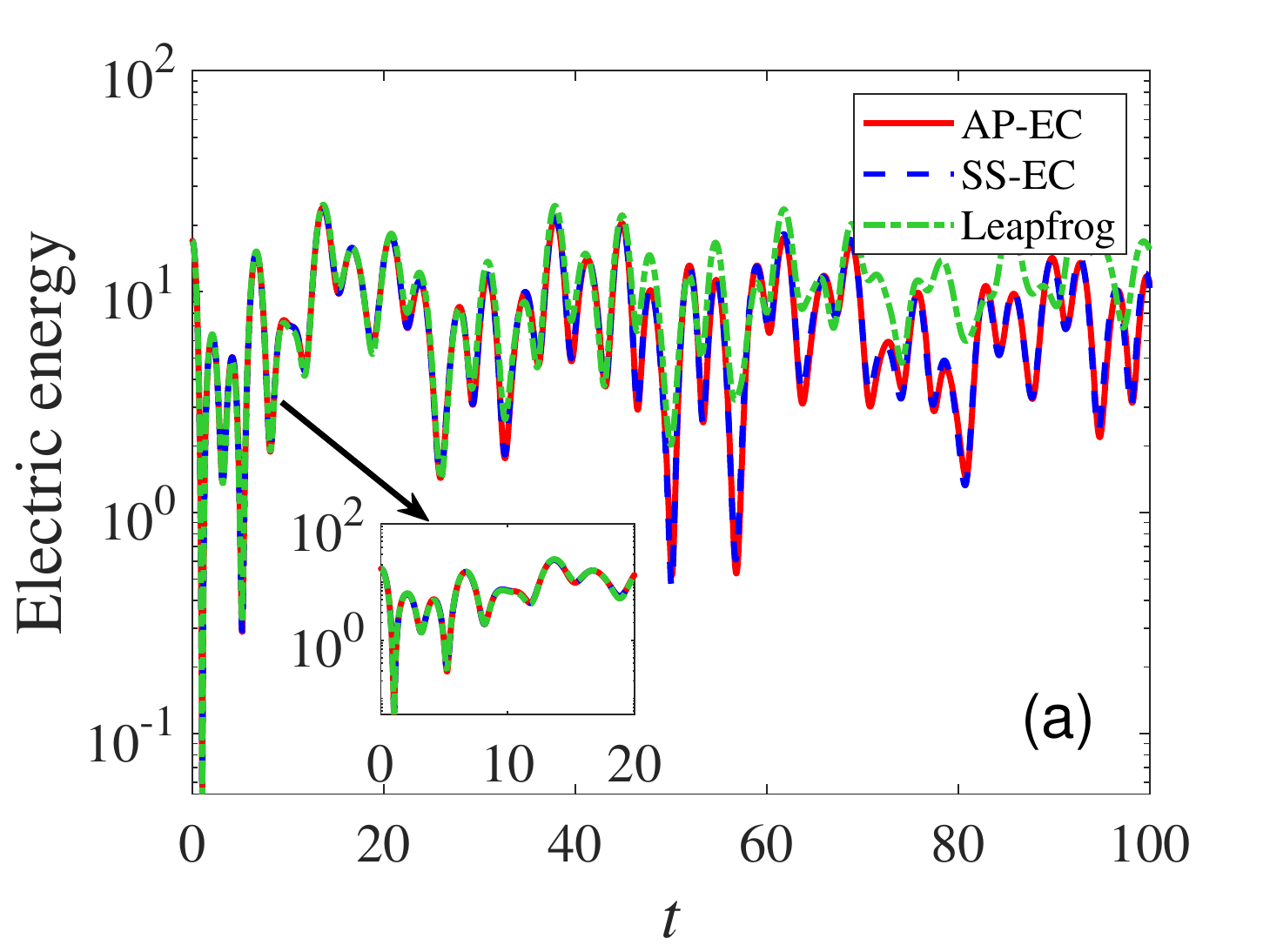}
    \includegraphics[width=0.3\textwidth]{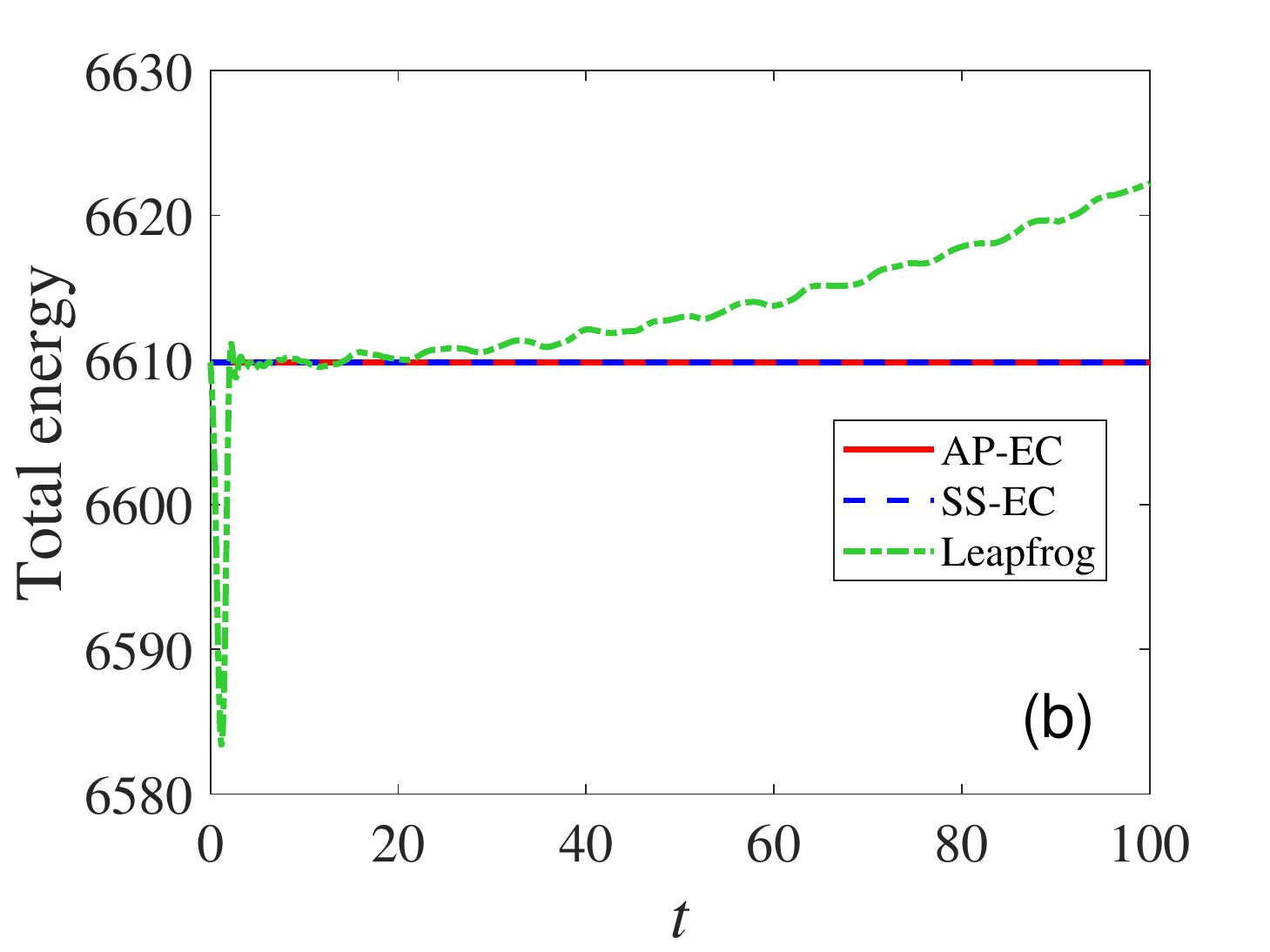}
    \includegraphics[width=0.3\textwidth]{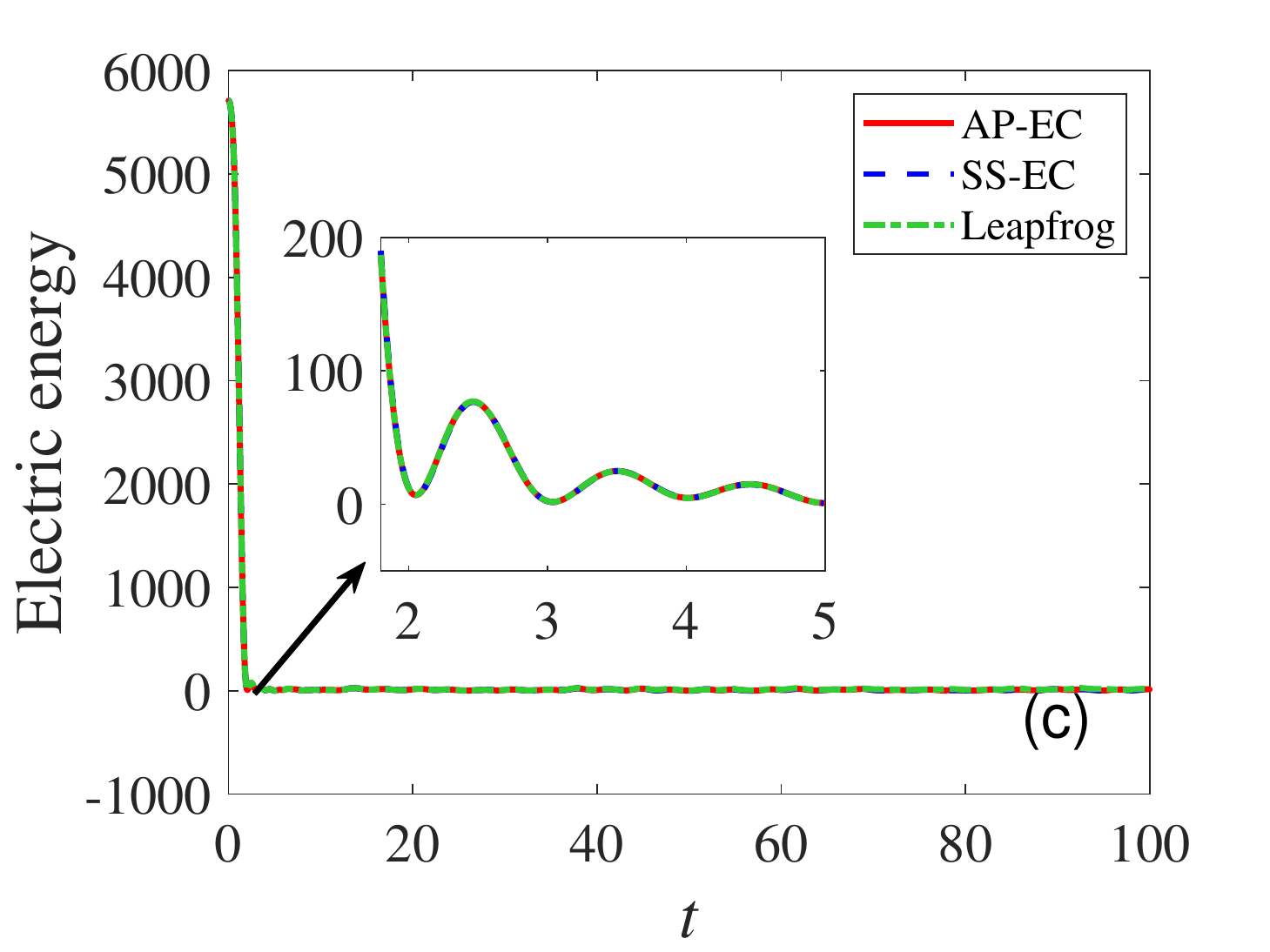}
    \caption{Bump-on-tail instability. $\lambda=1$, $\Delta t=0.01$. (a)~Electric energy of $x$-component; (b)~Total energy; (c)~Electric energy.}
    \label{fig: bump-on-tail-l1-dt001-k03-10000-g0}
\end{figure}

\begin{figure}[!htb]
    \centering
    \includegraphics[width=0.3\textwidth]{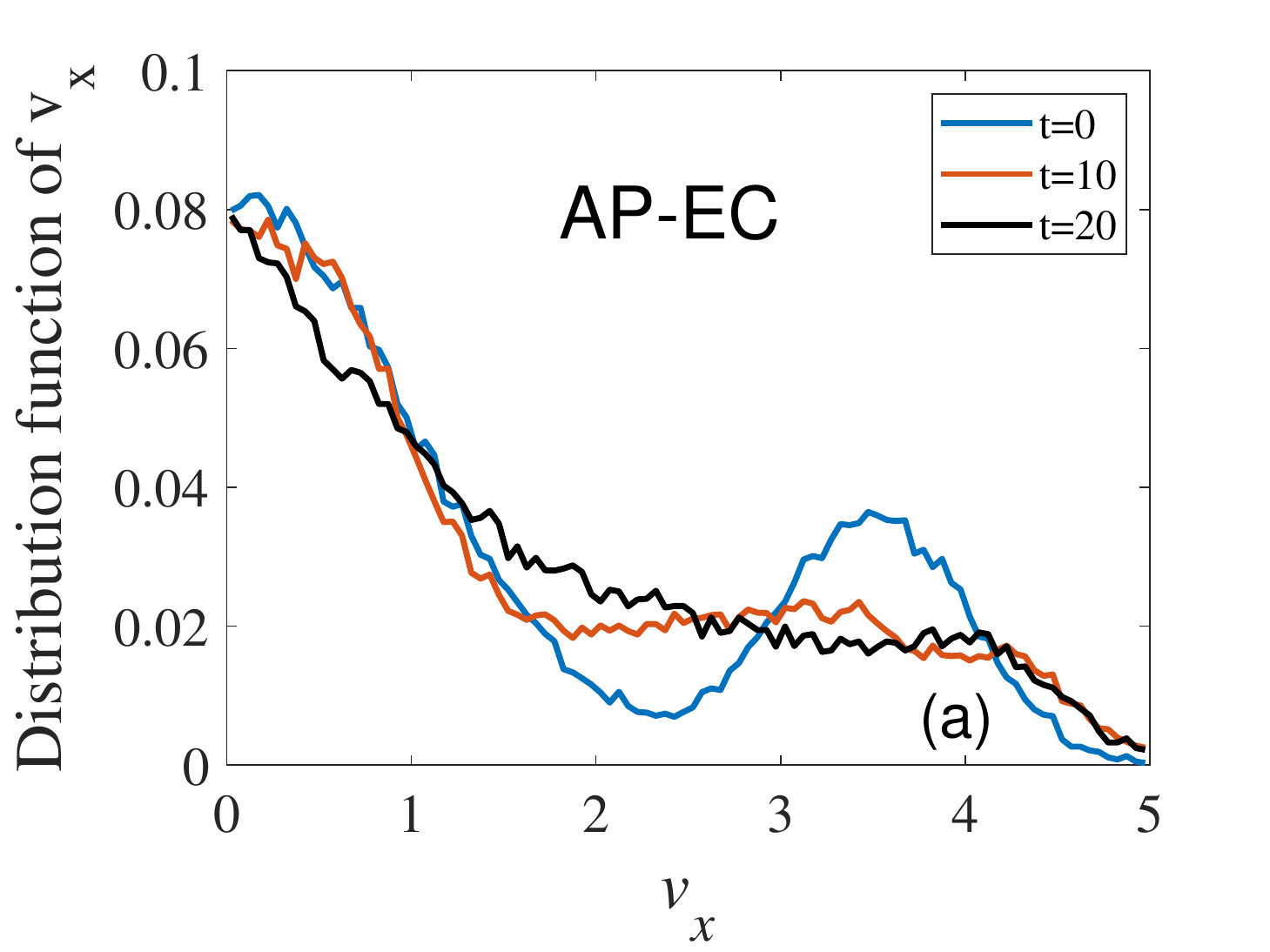}
    \includegraphics[width=0.3\textwidth]{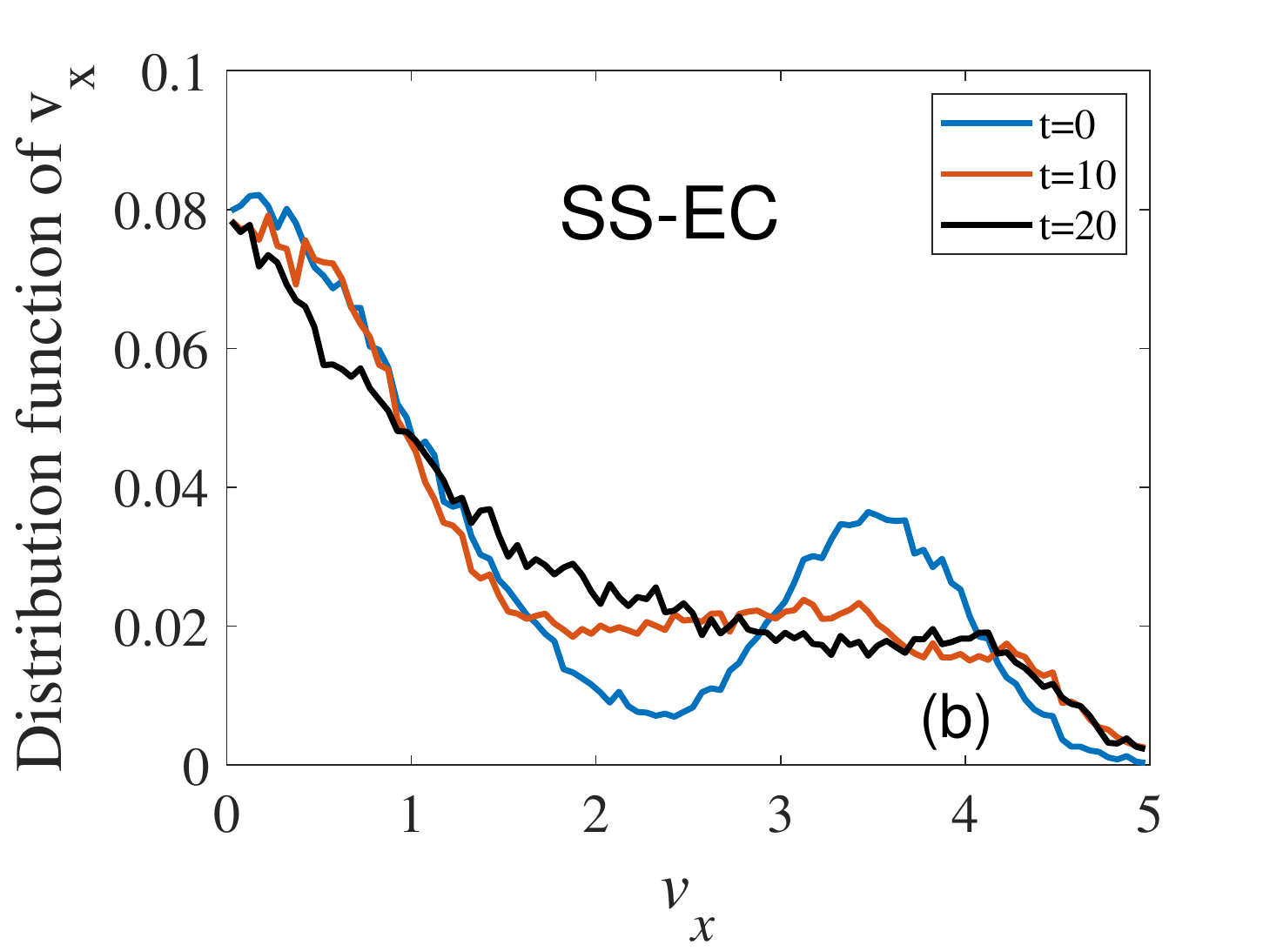}
    \includegraphics[width=0.3\textwidth]{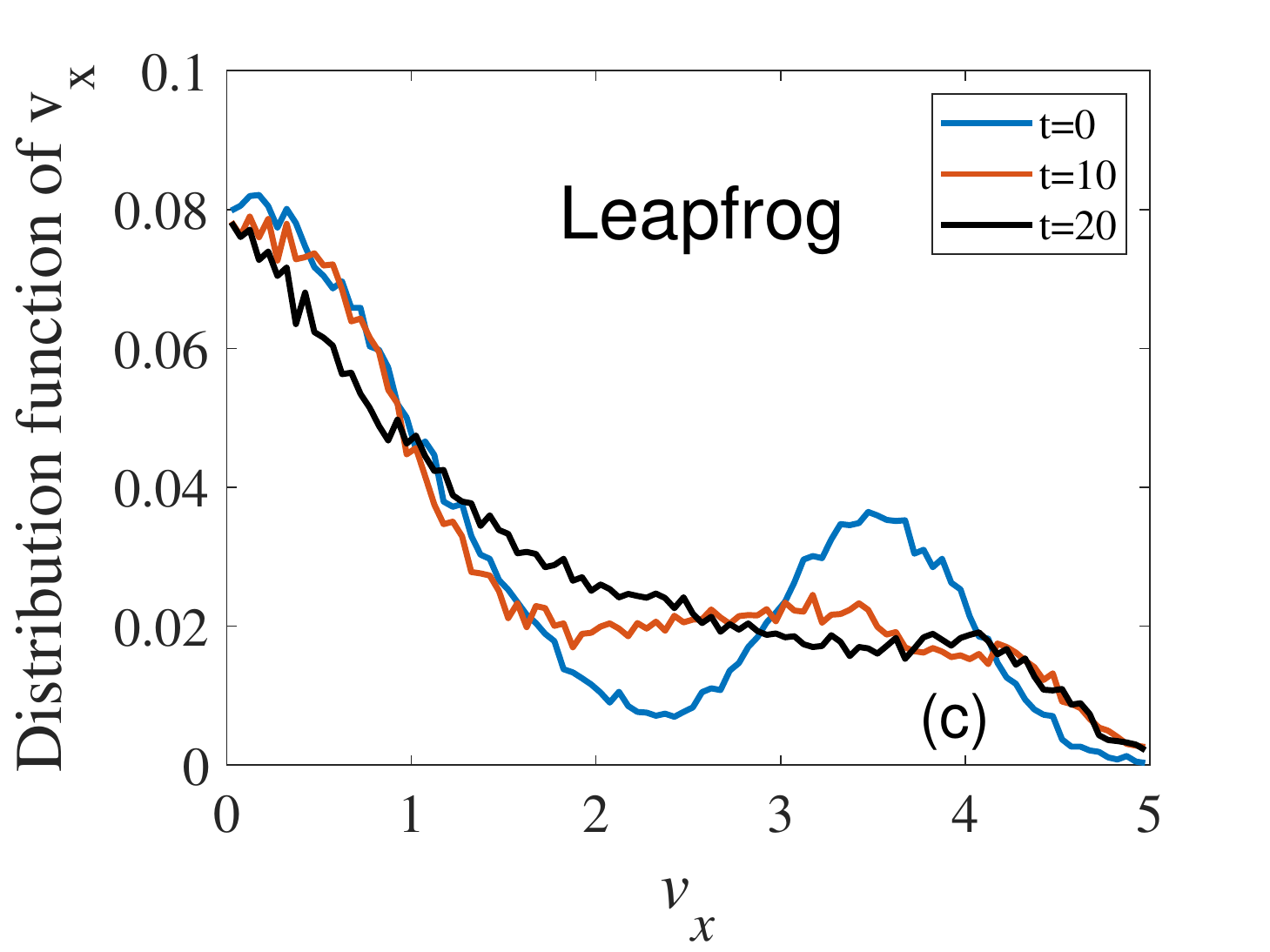}
    \caption{Distribution function of $v_{x}$ at $t=0$, $t=10$, $t=20$.}
    \label{fig: bump-on-tail-l1-dt001-k03-10000-g0-distribution-vx}
\end{figure}

Fig. \ref{fig: bump-on-tail-l1-dt001-k03-10000-g0-Ex} displays the profiles of electric field $\bs E_{x}$ of the AP-EC, SS-EC and leapfrog, respectively, where snapshots are saved at $t=0$, $t=10$, $t=20$ and $t=30$. Results imply that the AP-EC and SS-EC produce similar electric fields, while the leapfrog with $\Delta t=0.01$ gives different electric fields, in accordance with the results of Fig. \ref{fig: bump-on-tail-l1-dt001-k03-10000-g0}(a). Additional calculations show that $\bs{E}_{x}$ of leapfrog gets closer to that of the AP-EC and SS-EC with gradually decreased time steps.


\begin{figure}[!htb]
    \centering
    \includegraphics[width=0.95\textwidth]{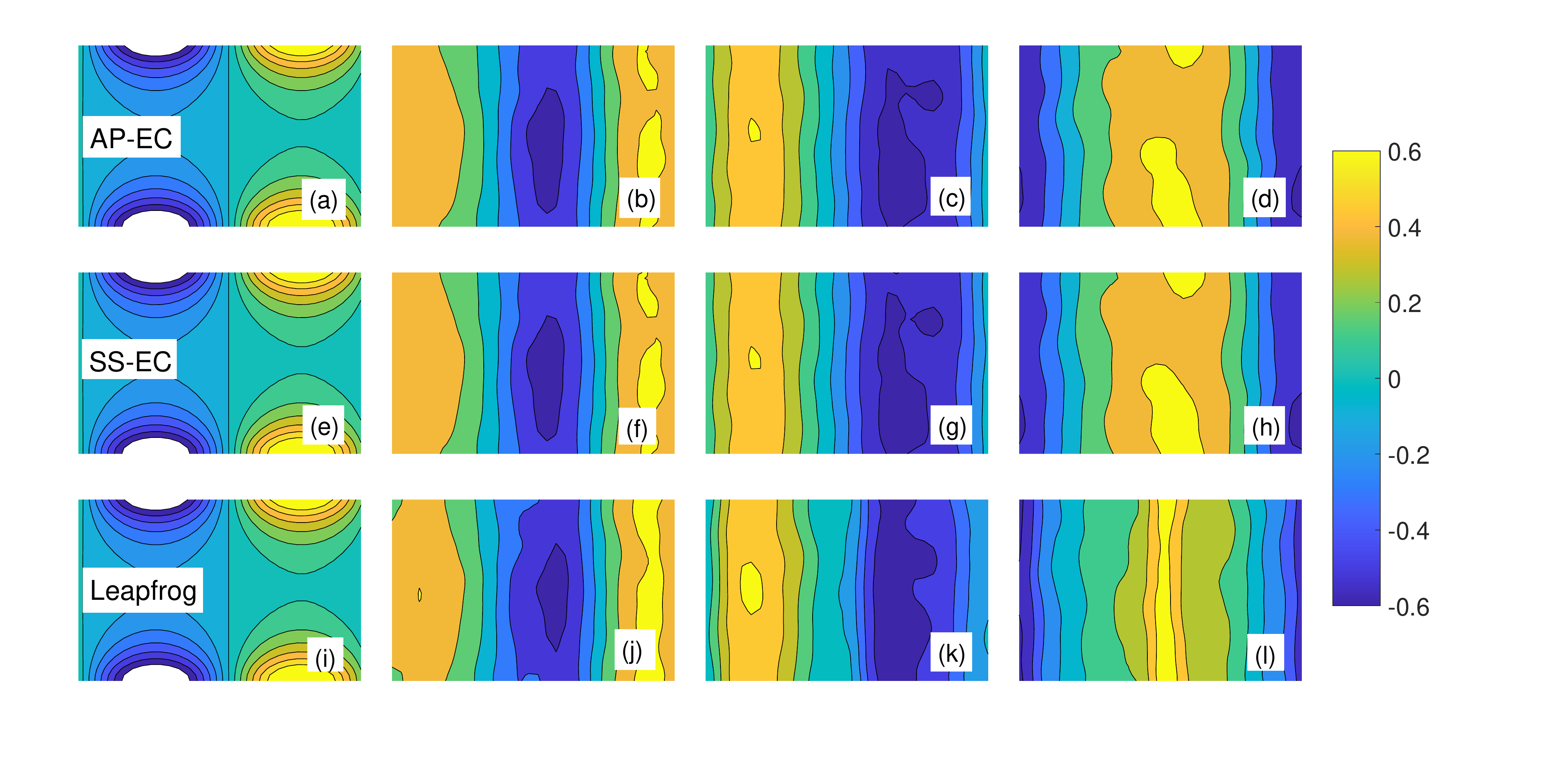}
    \caption{$\lambda=1$, $\Delta t=0.01$. Contour plot of electric field $\bs E_{x}$ at $t=0$, $10$, $20$ and $30$~(each column for one moment, respectively). (a-d): AP-EC; (e-h): SS-EC; (i-l): leapfrog.}
    \label{fig: bump-on-tail-l1-dt001-k03-10000-g0-Ex}
\end{figure}


\section{Conclusions}
\label{sec: conclusion}

In this paper we have developed an efficient energy-conserving implicit PIC method for approximating the high-dimensional electrostatic Vlasov system.  A constrained Vlasov-Amp\`ere system, reformulated from the original Vlasov-Poisson system, is introduced  as  the cornerstone of the proposed scheme. In order to properly discretize the solenoidal and irrotational fields involved in the VA system, a structure-preserving Fourier method is present, which exactly preserves the divergence-/curl-free constraints. With the help of the reformulated system and the novel structure-preserving Fourier method, a fully-implicit scheme based on time-centered temporal discretization is then proposed, which satisfies the discrete energy conservation regardless of the time step sizes. To accelerate the convergence of the resultant nonlinear system, an asymptotic-preserving preconditioner stemming from the generalized Ohm's law is employed. It can be viewed as a linearized and uniform approximation of the CN scheme for various Debye lengths. Together with a preconditioned Anderson-acceleration algorithm, the proposed fully-implicit scheme is robust and computationally efficient. In order to further reduce the computational cost, an energy-conserving Strang operator splitting method, which decouples the evolution of particle positions from the solution of particle velocities and electromagnetic fields, is further proposed with the help of the structure-preserving Fourier discretizations.

We have tested the proposed methods and compared their performance with the classical leapfrog method,  using extensive benchmark tests such as the Landau damping, the two-stream instability and the bump-on-tail instability. We
have shown that the asymptotic-preserving preconditioner has the merits of improved robustness and efficiency, and the proposed fully-implicit method generates physically accurate results for various time step sizes and Debye lengths, thus is more suitable for simulating complex plasmas with multiple physical scales. 

\section{Acknowledgement}
Z. Xu acknowledges the support from the NSFC (grant No. 12071288).
Z. Yang acknowledges the support from the NSFC (No. 12101399) and the Shanghai Sailing Program (No. 21YF1421000).
This work is also funded by the Strategic Priority Research Program of Chinese Academy of Sciences (grant Nos. XDA25010402 and XDA25010403).

\begin{appendices}
\label{appendix}

\section{The Anderson-accelerated iteration}

The nonlinear problem $\bs {\mathcal F}(\bs \xi)=\bs 0$ can be recast into an equivalent fixed-point problem
$
\bs \xi=\mathcal{T}(\bs \xi):= \bs {\mathcal F}(\bs \xi)+\bs \xi.
$
The Anderson-acceleration algorithm regarding this fixed-point problem is summarized in Algorithm \ref{alg: AAalg}, which is transferable for all fixed-point problems.

\begin{breakablealgorithm}\label{alg: AAalg}
\caption{Anderson-acceleration algorithm}
\centering
\begin{algorithmic}[1]
\Require Initial guess $\bs{\xi}^{(0)}$, truncation depth $m\geq 1$, maximum iteration $K\geq 1$ and tolerance $0<\varepsilon \ll 1$.
\Ensure The numerical solution $\bs{\xi}$.
\State Set $k=1$,  $\bs{ {\xi}}^{(1)}=\mathcal{T}(\bs{\xi}^{(0)})$, $\bs {\mathcal{F}}(\bs \xi^{(1)})=\mathcal{T}(\bs{\xi}^{(1)})-\bs{ {\xi}}^{(1)}$.
\While{$k<K$ $\&$ $\|\bs {\mathcal F}(\bs \xi^{(k)})\|\leq \varepsilon$}
\State Set $m_{k}=\min \{m,k\}$.
\State Set $F^k=\big(   \bs {\mathcal{F}}(\bs \xi^{(k-m_k)}),\cdots, \bs {\mathcal{F}}(\bs \xi^{(k)}) \big)$.\vspace{4pt}
\State Find a column vector $\bs{\gamma}^{(k)}=(\gamma_0^{(k)},\gamma_1^{(k)},\cdots,\gamma_{m_k}^{(k)})^T$ such that
\begin{equation}
    \displaystyle\min\limits_{\bs \gamma^{(k)}}\Big \|\sum\limits_{j=0}\limits^{m_{k}}     \gamma_j^{(k)}   \bs {\mathcal{F}}\big(\bs \xi^{(k-m_k+j)}\big)    \Big \|_2,\;\text{ s.t.}\,\sum\limits_{j=0}\limits^{m_{k}}\gamma_{j}^{(k)}=1.
\end{equation}
\State Set $\bs{\xi}^{(k+1)}=\sum\limits_{j=0}\limits^{m_{k}}     \gamma_j^{(k)}   \bs {\mathcal{T}}(\bs \xi^{(k-m_k+j)})$,\;\;$\bs {\mathcal{F}}(\bs \xi^{(k+1)})=\mathcal{T}(\bs{\xi}^{(k+1)})-\bs{ {\xi}}^{(k+1)}$.
\State Set $k=k+1$.
\EndWhile
\State $\bs{\xi}=\bs{\xi}^{(k+1)}$.
\end{algorithmic}
\end{breakablealgorithm}

\end{appendices}

\bibliographystyle{abbrv}
\bibliography{main}

\end{document}